\documentclass[english,a4paper,12pt]{article}
\usepackage[utf8]{inputenc}
\usepackage[T1]{fontenc}
\usepackage{titling}
\usepackage{amsmath, amssymb}
\usepackage{pgf}
\usepackage{amsfonts}
\usepackage{dsfont}
\usepackage{mathrsfs}
\usepackage{mathabx}
\usepackage{stmaryrd}
\usepackage{makeidx}
\usepackage{amsbsy}
\usepackage{amsthm}
\usepackage{color}
\usepackage{fullpage}
\usepackage{enumitem}
\usepackage[vcentermath]{youngtab}
\usepackage{marginnote} 
\usepackage{mdframed}	
\newmdenv[topline=false, bottomline=false, skipabove=\topsep, skipbelow=\topsep]{siderules}
\usepackage{tikz}
\usetikzlibrary{decorations.text,calc,arrows.meta}
\usepackage[vcentermath]{youngtab}
\usepackage{tkz-euclide}
\usetkzobj{all}
\usepackage{pgfplots}
\usetikzlibrary{decorations.shapes}
\usetikzlibrary{positioning}
\usetikzlibrary{patterns}
\usetikzlibrary{shapes.geometric}
\usetikzlibrary{decorations.pathmorphing}
\usetikzlibrary{arrows.meta}
\tikzset{snake it/.style={decorate, decoration=snake}}
\usetikzlibrary{arrows,shapes,positioning}
\usetikzlibrary{decorations.markings}
\tikzstyle arrowstyle=[scale=1]
\tikzstyle directed=[postaction={decorate,decoration={markings,mark=at position .65 with {\arrow[arrowstyle]{stealth}}}}]
\tikzstyle reverse directed=[postaction={decorate,decoration={markings,mark=at position .65 with {\arrowreversed[arrowstyle]{stealth};}}}]

\newtheorem{theorem}{Theorem}
\newtheorem{corollary}{Corollary}

\newtheorem{proposition}{Proposition}
\newtheorem{lemma}{Lemma}
\newtheorem{remark}{Remark}

\newtheorem{assumption}{Assumption}

\def\bD{{\mathbb D}}

\def\bN{{\mathbb N}}

\def\bS{{\mathbb S}}

\def\bZ{{\mathbb Z}}

\def\gA{{\mathfrak A}}
\def\gB{{\mathfrak B}}

\def\gM{{\mathfrak M}}

\newcommand{\A}{\mathfrak{A}}

\newcommand{\ca}[1]{{\cal #1}}
\newcommand{\ben}{\begin{equation}}
\newcommand{\een}{\end{equation}}
\def\bena{\begin{eqnarray}}
\def\eena{\end{eqnarray}}

\def\cD{{\ca D}}

\def\cH{{\ca H}}

\def\cK{{\ca K}}

\def\cP{{\ca P}}

\def\cV{{\ca V}}

\renewcommand{\H}{\mathcal{H}}

\newcommand{\K}{\ensuremath{\mathcal{K}}}

\def\1{{\mathds{1}}}

\def\supp{{\mathrm{supp} \,}}

\def\GU{{\mathrm{SU}(1,1)}}

\def\Ad{{\mathrm{Ad}}}

\newcommand{\dd}{{\rm d}}
\newcommand{\tr}{\operatorname{Tr}}
\newcommand{\bint}{\operatorname{\int \!\!\!\!\! -- \, }}

\renewcommand{\log}{\operatorname{ln}}

\renewcommand{\epsilon}{\varepsilon}

\newcommand{\RR}{\mathbb{R}}
\newcommand{\CC}{\mathbb{C}}
\renewcommand{\SS}{\mathbb{S}}
\newcommand{\ZZ}{\mathbb{Z}}

\newcommand{\half}{\tfrac{1}{2}}

\begin{document}
\title{On the modular operator of mutli-component regions in chiral CFT}

	\author{
	Stefan Hollands$^{1}$\thanks{\tt stefan.hollands@uni-leipzig.de} 
		\\ \\
{\it ${}^{1}$Institut f\" ur Theoretische Physik,
Universit\" at Leipzig, }\\
{\it Br\" uderstrasse 16, D-04103 Leipzig, Germany} \\
	}
\date{\today}
	
\maketitle

\begin{abstract}
We introduce a new approach to find the Tomita-Takesaki modular flow for multi-component regions in general chiral conformal field theory. 
Our method is based on locality and analyticity of primary fields as well as the so-called Kubo-Martin-Schwinger (KMS) condition. 
These features can be used to transform the problem to a Riemann-Hilbert problem on a covering of the complex plane cut along the regions, 
which is equivalent to an integral equation for the matrix elements of the modular Hamiltonian. Examples are considered.
\end{abstract}

\section{Introduction}

The reduced density matrix of a subsystem  induces an intrinsic internal dynamics called the ``modular flow''. The flow is non-trivial only for non-commuting 
observable algebras -- i.e., in quantum theory -- and depends on both the subsystem and the given state of the total system. It has been subject to much attention in theoretical physics in recent times because it is closely related to information theoretic concepts. As examples for some topics such as 
Bekenstein bounds, Quantum Focussing Conjecture, c-theorems, holography we mention \cite{Longo1,bousso1,bousso2,Casini2,Hubeny}. In mathematics, the modular flow has played an important role in the study of operator algebras through the work of Connes, Takesaki and others, see \cite{Takesaki} for an encyclopedic account.

It has been known almost from the beginning that the modular flow has a geometric nature in local quantum field theory 
when the subsystem is defined by a spacetime region of a simple shape such as an interval in chiral conformal field theory (CFT) \cite{bisognano, hislop, longo3}: it is the 1-parameter group 
of M\" obius transformations leaving the interval fixed. For more complicated regions,
important progress was made only much later in a pioneering work by Casini et al. \cite{Casini0}, who were able to determine the flow for multi-component regions for a chiral half of free massless fermions in two dimensions. Recently in \cite{Casini1} they have generalized their method to the conformal theory of a chiral $U(1)$-current. 
Unfortunately, the method by  \cite{Casini0,Casini1}, as well as all other concrete methods known to the author, is based in an essential way on special properties of free quantum field theories. The purpose of this paper is to develop methods that could give a handle on the problem in general chiral CFTs, i.e. the left-moving half of a CFT on a (compactified) lightray in 1+1 dimensional Minkowski spacetime, and to make some of the constructions in the literature rigorous by our alternative method. 

Consider a (possibly mixed) state in the chiral CFT, described by a density matrix $\rho$. Typical states of interest are the vacuum $\rho = |\Omega_0\rangle \langle \Omega_0|$, 
or a thermal state $\rho = e^{-\beta L_0}/\tr e^{-\beta L_0}$. Given a union $A=\cup_j (a_j,b_j)$ of intervals of the (compactified) lightray, we can consider its reduced density 
matrix $\tr_{A'} \rho = \rho_A$, where $A'$ is the complement of $A$. For the purposes of this discussion, we restrict to the vacuum state, although in the main part, 
thermal states will play a major role as well. If $\phi(x)$ is a primary field localized at $x \in A$, the ``modular flow'' is the Heisenberg time evolution 
$\rho_A^{it} \phi(x) \rho^{-it}_A$. The object $\rho_A$ is not actually well-defined in quantum field theory, but the modular flow is. Below we will use the rigorous 
framework of Tomita-Takesaki theory in our construction, but for pedagogical purposes, we here pretend that $\rho_A$ exists. Formally, the Hilbert space $\H$ splits as 
$\H_A \otimes \H_{A'}$ and if $\rho$ is pure, then $\rho_A$ is formally a density matrix on $\H_A$. Its -- equally formal -- logarithm $H_A = \log \rho_A$ is 
called the modular Hamiltonian in the physics literature. 

 In mathematical terms, the quantity which is well defined is the operator $\Delta = \rho_A \otimes \rho_{A'}^{-1}$. For $x \in A$, we can then also write 
 $\rho_A^{it} \phi(x) \rho^{-it}_A = \Delta^{it} \phi(x) \Delta^{-it}$ and $\log \Delta = H_A \otimes 1_{A'}-1_A \otimes H_{A'}$. Furthermore, one can 
 write 
\ben
\label{GNS}
\tr(\phi(x)\rho^{it}_A\phi(y)\rho^{1-it}_A) = \langle \Omega_0| \phi(x) \Delta^{it} \phi(y) \Omega_0 \rangle. 
\een
and since the conformal primaries generate the full Hilbert space (mathematically, the Reeh-Schlieder theorem), we see that we knowledge of this 
quantity for all primaries $\phi$ suffices, in principle, to determine all 
matrix elements of $\Delta^{it}$, hence the operator itself, hence the flow. Alternatively, to know the generator of the flow, it suffices to know
$\langle \Omega_0| \phi(x) (\log \Delta) \phi(y) \Omega_0 \rangle$. It is those types of quantities which we will study in this paper. 

Our main trick is the following observation and it variants. For $s>0$ and fixed $y \in A$, define a function of $x$ on the complex plane cut along the intervals $A$,
\ben
\label{Fdef11}
F(s, x, y) = 
\begin{cases}
\langle \Omega_0 | \phi(x) [1-e^s (1-\Delta)^{-1}]^{-1} \phi(y) \Omega_0 \rangle & \text{if $\Im(x)<0$,}\\
\langle \Omega_0 | \phi(y) [1-e^s(1-\Delta^{-1})^{-1}]^{-1} \phi(x) \Omega_0 \rangle & \text{if $\Im(x)>0.$}
\end{cases}
\een
Then not only do the usual properties of CFTs imply that this function is holomorphic on the mutliply cut plane, but we also know 
its jumps across the cuts, given by the functional equation
\ben\label{jump001}
(1-e^s) F(s, x-i0, y) - F(s, x+i0, y) = \langle \Omega_0 | [ \phi(x), \phi(y) ] \Omega_0 \rangle.
\een
The commutator on the right side is given by a sum of $\delta$-functions and their derivatives by locality. 
We also prove certain further general properties of this function such as the degree of divergences as $x$ approaches $y$ or
any boundary of a cut which depend on the conformal dimension of $\phi$. 
Using this functional equation and a standard contour argument appearing frequently in the study of Riemann-Hilbert type problems, we then 
obtain a linear integral equation for $F$ of Cauchy-type, see cor. \ref{cor:3} (in the case of bosonic fields).
The desired matrix 
elements of the modular Hamiltonian are related by the integral 
\ben
\langle \Omega_0| \phi(x) (\log \Delta) \phi(y) \Omega_0 \rangle = \int_0^\infty \dd s \, F(s,x,y) .
\een
A variant of this method also works for fermionic fields and for 
thermal states where the corresponding function $F$ lives on a torus cut along $A$ and satisfies a corresponding integral 
equation, see cor. \ref{cor:2}. 

The basis of our method is in some sense an old trick in quantum statistical mechanics. Consider a statistical operator $\rho$. The expectation functional acting an observable $X$ is $\omega(X)=\tr(X\rho)$ and the modular flow acting on an observable $X$ is, by definition, $\sigma^t(X) = \rho^{it} X \rho^{-it}$. For observables $X,Y$, consider the function $\varphi_{X,Y}(t) = \omega(X\sigma^t(Y)) = \tr(X\rho^{it}Y\rho^{1-it})$. Since $\rho$ is a positive operator, one expects this function to be analytic inside the strip $\{ t \in \CC |  -1 < \Im(t) < 0 \}$. The values at the two boundaries of the strip are evidently related by the functional equation
\ben
\varphi_{X,Y}(t-i) = \varphi_{Y,X}(-t). 
\een
This functional equation is called the ``KMS-condition.'' Its fundamental importance was first understood in \cite{HHW}. 

Note that on the right side, $X,Y$ appear in opposite order. Thus, if we have information about their commutator, we can sometimes get a closed equation for $\varphi_{X,Y}(t)$ or related quantities. This is for instance the case for the ideal quantum Bose or Fermi gas, where  $\rho = Z^{-1} e^{-\beta H}, H=\sum E_k^{} a^*_{\bf k} a_{\bf k}^{} $, the modular flow is just the Heisenberg time evolution
with `time' parameter $-\beta t$, and where one 
takes $Y=a_{\bf k}, X=a_{\bf p}^*$. Using the commutators and $\varphi_{\bf p,k}(t)=e^{-it\beta E_k}\varphi_{\bf p,k}(0)$, the  KMS condition is thereby 
equivalent in the case of bosons to the condition
\ben
(1-e^{-\beta E_k})\varphi_{\bf p,k}(0) = e^{-\beta E_k} \delta^3({\bf k}-{\bf p}) .
\een
In this way, one can easily derive the standard Bose-Einstein (or Fermi-Dirac-) formula for the 2-point function. In our problem, these ideas are modified and applied 
to the reduced density matrix $\rho_A$ of a CFT.

This paper is organized as follows. In secs. \ref{sect1}, \ref{sect3}, we review basic notions from operator algebras, Tomita-Takesaki theory, and the operator algebraic 
approach to CFT (conformal nets) in order to make the paper self-contained. In secs. \ref{sect2}, \ref{sec:sect5}  we introduce our method and study several examples. 
We conclude in sec. \ref{sec:Conclusion}.
Some conventions for elliptic functions are described in the appendix.

\medskip
\noindent
{\bf Notations and conventions:} Gothic letters $\gA, \gM, \dots$ denote $*$-algebras, usually v. Neumann algebras. Calligraphic letters $\H, \K, \dots$ denote linear spaces, always assumed to be separable. The inverse temperature $\beta$ and modular parameter $\tau$ are related by $-2\pi i \tau = \beta$. The branches of 
$\log z$ and $z^\alpha$ are taken along the negative real axis. $\bS=\{z \in \CC \mid |z|=1\}$ denotes the unit circle, $\bD^\pm$ its interior/exterior.

\medskip
\noindent
{\bf Note added in proof:} After this preprint was submitted, it was pointed out to us by the authors of \cite{blanco} that one of our calculations related to thermal states 
contained an error, creating a tension between some of our results and those by \cite{blanco}, see also \cite{fries}. We are grateful to these authors for making us aware of this issue, which has been fixed 
in the current version.

\section{Review of modular theory}\label{sect1}

\subsection{Modular flow}

For the convenience of the unfamiliar reader we review the basic elements of modular ($=$ Tomita-Takesaki-) theory; detailed references are \cite{Takesaki,Bratteli}. 
Connections to quantum information theory are described in \cite{sanders_2}. An exposition directed towards a theoretical physics audience is \cite{witten}. 

The notion of modular flow is embedded into the theory of v. Neumann algebras. Such an algebra, $\gM$, can be defined as 
a complex linear space of bounded operators on some Hilbert space\footnote{We always assume that $\H$ is separable.} $\H$ that is 
closed under taking products, adjoints (denoted by $*$). Such limits are understood in the so called ``weak'' topology, i.e. convergence of 
matrix elements.  It is common to denote by $\gM'$ the commutant, defined as the set of all bounded operators on $\H$ 
commuting with all operators in $\gM$.

To define the objects of main interest of the theory, one has to assume that $\gM$ is in ``standard form'', meaing: 
$\H$ contains a ``cyclic and separating'' vector for $\gM$, that is, a unit vector $|\Omega \rangle$
such that the set consisting of $X|\Omega\rangle$, $X \in \gM$ is a dense subspace of $\H$, and such that $X|\Omega\rangle=0$ implies $X=0$
for any $X \in \gM$. The point is that one can then consistently define the anti-linear Tomita operator $S$ on the domain 
${\mathcal D}(S) = \{ X|\Omega\rangle \mid X \in \gM\}$ 
by the formula
\ben
SX|\Omega \rangle = X^* |\Omega \rangle . 
\een
The cyclic property is needed in order that $S$ is densely defined, whereas without the separating property the definition would not be self-consistent. One can show that $S$ is a closable operator. This technical property guarantees that $S$ has a polar decomposition. It is customarily denoted by $S=J\Delta^\frac12$, where $J$ anti-linear and unitary and $\Delta$ self-adjoint and non-negative. Tomita-Takesaki theory is about the interplay between the operators $\Delta, J$ and the algebras $\gM, \gM'$.  The basic theorem is: 

\begin{enumerate}
\item[(i)]
$J$ exchanges $\gM$ with the commutant in the sense that 
$J \gM J = \gM'$. Furthermore, $J^2 = 1, J\Delta J = \Delta^{-1}$. 

\item[(ii)] 
The modular flow $\sigma^t(X) =  \Delta^{it} X \Delta^{-it}$ leaves $\gM$ and $\gM'$ invariant for all $t \in \RR$.

\item[(iii)]
From the vector $|\Omega\rangle$, one can define the state functional 
$
\omega(X) = \langle \Omega |X\Omega \rangle, \quad \omega: \gM \to \CC. 
$
It is positive and normalized (meaning $\omega(X^*X) \ge 0 \, \, \forall X \in \gM, \omega(1) = 1$), and invariant under the modular flow 
in the sense that $\omega \circ \sigma^t = \omega$ for all $t \in \RR$. The {\bf KMS-condition} holds: for all $X,Y \in \gM$, the bounded function 
\ben
\label{Fdef}
t \mapsto \varphi_{X,Y}(t) = \omega (X\sigma^t(Y)) \equiv \langle \Omega | X \Delta^{it} Y \Omega \rangle
\een
has an analytic continuation to the strip $\{ z \in \CC \mid -1 < \Im z < 0 \}$ with the property that its boundary value for $\Im z \to -1^+$ exists and
is equal to
\ben
\label{KMS}
\varphi_{X,Y}(t-i) = \omega (\sigma^t(Y) X).
\een
\end{enumerate}

A partial converse to (iii) is: If $\omega'$ is a normal (i.e. continuous in the weak$^*$-topology) positive linear functional on $\gM$, 
then it has a unique vector representative $|\Omega' \rangle$ in the natural cone 
$\cP^\sharp:=\{ Xj(X) |\Omega \rangle \mid X \in \gM\}$,
where $j(X) = JXJ$; in other words
$\omega'(X) = \langle \Omega' |X\Omega' \rangle$ for all $X \in \gM$. 

The objects $J,\Delta,\cP^\sharp$ depend on the algebra $\gM$ and the state $|\Omega\rangle$. 

\medskip
\noindent
{\bf Example 0}: Even tough Tomita-Takesaki theory is most interesting in the case of infinite dimensional 
v. Neumann algebras of types II, III, it helps with intuition to have in mind the finite dimensional case, i.e.
the type I$_n$' (algebra of $n$ by $n$ matrices). 
In this case, $\gM = M_n(\CC) \otimes 1_n$, which acts 
on the Hilbert space $\H = \CC^n \otimes \CC^n$. 
Evidently, the commutant is $\gM' = 1_n \otimes M_n(\CC)$. A vector $|\Omega\rangle$ in this Hilbert space 
is cyclic and separating if $|\Omega \rangle = \sum_{j=1}^n \sqrt{p_j} |j \rangle \otimes |j\rangle$ in some ON basis $\{|j\rangle\}$ and iff all $p_j > 0$, $\sum_{j=1}^n p_j=1$. 
The state functional $\omega$ can be written in this example in terms of the ``reduced density matrix''  
\ben
\rho_\omega = \sum_{j=1}^n p_j |j \rangle \langle j|, \quad \omega(X) = \tr_{\CC^n}(X \rho_\omega) \quad (X \in \gM).
\een
In fact, any positive normalized state functional $\omega'$ arises from a unique reduced density matrix $\rho_{\omega'}$ in this way. 
It is easy to go through the definition of $\Delta, J$ via $S$ giving for instance that
\ben
\Delta^\frac12 = \rho^{\frac12}_\omega \otimes \rho^{-\frac12}_\omega .  
\een
Therefore, the modular flow is $\sigma^t(X) = \rho^{it}_\omega X \rho^{-it}_\omega$. The ``modular Hamiltonian'' is defined as 
the self-adjoint operator $\log \Delta$. In our example, therefore,
$
\log \Delta = \log \rho_\omega \otimes 1_n - 1_n \otimes \log \rho_\omega , 
$ 
where the first term belongs to $\gM$ and the second to $\gM'$. It is important to stress that the split of $\log \Delta$ into a part from $\gM$ and one from $\gM'$ 
is impossible for general v. Neumann algebras, in particular for the type III$_1$-factors appearing in quantum field theories\footnote{The possibility of making the split implies that 
$\sigma^t$ is inner, i.e. can be written as $\sigma^t(X) = U(t) X U(t)^*$ for unitaries $U(t)$ in $\gM$. One characterization of type $III$ v. Neumann algebras 
is that $\sigma^t$ precisely cannot be inner for any normal state $\omega$.}. Therefore, apart from trivial cases, the object $\log \rho_\omega$, hence the reduced 
density operator $\rho_\omega$ itself, does not exist. On the other and, $\log \Delta$ and $\omega$ always exist. 
We will make sure to work with these well-defined objects in our setting.

\medskip

Sometimes, a state $\omega$ is only given as an abstract (weakly continuous) expectation functional on an abstract\footnote{We mean a $C^*$-algebra with a preferred ``folium'' of normal states, see \cite{Takesaki}. In particular, it is not assumed that $\gM$ is a priori represented by bounded linear operators on some Hilbert space.} v. Neumann algebra $\gM$. Then one can perform the basic but very important {\bf GNS construction} in order to obtain a Hilbert space in which the state is represented by a vector. 

The starting point of this construction is the simple observation that the algebra $\gM$ itself, as a linear space, always forms a representation $\pi$ by left multiplication, i.e. $\pi(X) Y \equiv XY$. To equip this representation with a Hilbert space structure, it is natural to define $\langle X|Y\rangle = \omega(X^* Y)$, but this will in general lead to non-zero vectors with vanishing norm, unless $\omega$ is separating. Introduce ${\mathfrak J}_\omega = \{ X \in \gM \mid \omega(X^*X) = 0\}$. By the Cauchy-Schwarz inequality, $|\omega(X^*Y)| \le \omega(X^*X)^{1/2} \omega(Y^*Y)^{1/2}$, we have ${\mathfrak J}_\omega = \{ X \in \gM \mid \forall Y\in \gM, \omega(Y^*X) = 0\}$, so it is a closed linear subspace and a left ideal of $\gM$ containing precisely the null vectors. We can then define $\H_\omega = \gM/{\mathfrak J}_\omega$ and complete it in the induced inner product. The left representation induces a representation on $\H_\omega$ which is called $\pi_\omega$. It is the desired GNS-representation. The vector $|\Omega_\omega \rangle \in \H_\omega$ representing $\omega$ is simply the equivalence class of the unit operator, $1$. It is by construction ``cyclic'' in the sense that the set $\pi_\omega(\gM)|\Omega_\omega \rangle$ is dense in $\H_\omega$. The vector is standard if $\omega(X^*X)$ implies $X=0$ (meaning ${\mathfrak J}_\omega = \{0\}$), in which case we say that it is faithful.

\section{Review of chiral CFTs}
\label{sect3}

\subsection{Conformal nets on the real line (lightray)}
One way to formalize the structure of chiral conformal quantum field theories (CFTs) is via nets of operator algebras.  A chiral conformal field theory is associated with one  lightray. It is given abstractly by an assignment of 
an algebra of operators $\gA(I)$ with each open interval $I = (a,b) \subset \RR$ of the this lightray. 

This assignment is called a {\bf conformal net} if it obeys the following rules (see \cite{haag} for a general introdution to algebraic quantum field theory 
and e.g. \cite{frohlich,longo2} for conformal nets):

\begin{enumerate}
\item[a1)] (Isotony) The algebras $\gA(I)$ are v. Neumann algebras acting on a common Hilbert space $\H$. If $I \subset J$ are intervals, then 
$\gA(I) \subset \gA(J)$. 

\item[a2)] (Causality) Setting $I' = \RR \setminus [a,b]$ if $I=(a,b)$, we have $\gA(I') \subset \gA(I)'$, i.e. observables from disjoint intervals commute. 

\item[a3)] (Covariance) On $\H$, there is a unitary representation $g \mapsto U(g)$ of the group ${\rm SL}(2, \RR)/\{\pm 1\}$. If we let elements $g = 
\left( 
\begin{matrix}
a & b \\
c & d
\end{matrix}
\right)
$ of this group act locally on 
$\RR$ by fractional transformations $g(x) = \frac{ax + b}{dx + c}$, then it is assumed that 
$U(g) \gA(I) U(g)^* = \gA(g(I))$ for all intervals $I$ and $g \in {\rm SL}(2, \RR)$ such that $g(x)$ is well defined for all $x \in I$. We also use the notation
\ben
\alpha^g(X) = \Ad_{U(g)}(X) \equiv U(g) X U(g)^*. 
\een

\item[a4)] (Spectrum) The representation $g \mapsto U(g)$ is strongly continuous. The infinitesimal generator $P$ of translations ${\rm tra}_t(x) = x+t$, i.e. $P = -i\frac{\dd}{\dd t} U({\rm tra}_t)|_{t=0}$ has non-negative spectrum.   

\item[a5)] (Vacuum) There is a unique (unit) vector $|\Omega_0 \rangle \in \H$ such that $U(g) |\Omega_0 \rangle = |\Omega_0 \rangle$.
The corresponding state functional will be called $\omega_0(X) = \langle \Omega_0 | X \Omega_0 \rangle$ throughout. The vacuum 
should be cyclic for $\bigvee_I \gA(I)$, the v. Neumann algebra generated by all intervals. 
\end{enumerate}

The algebra of observables associated with the union of $p$ open intervals  with disjoint closures,
\ben
\label{arcs}
A = \bigcup_{j=1}^p (a_j, b_j) \subset \RR \quad \text{or $\bS$ below,}
\een
where each $I_j = (a_j, b_j)$ is an interval of $\RR$ (or arc of the circle $\bS$ below), is defined to be 
\ben\label{generatedvN}
\gA(A) = \bigvee_{i=1}^p \gA(I_i), 
\een
where the symbol $\vee$ means the v. Neumann algebra that is generated by the algebras for the individual arcs/intervals.

\subsection{Conformal nets on the circle (compactified lightray)}

If we want to insist on a global action of the M\" obuis group ${\rm SL}(2, \RR)/\{\pm 1\}$ on the net, we must pass from the light ray to a 
compactified lightray, i.e. the circle. The compactification proceeds via the  Caley transformation $C: \bS \setminus \{+1\} \to \RR, C(x) = -i(x+1)/(x-1)$, and 
under this transformation intervals get mapped to arcs of the circle. The Caley transform intertwines the action of ${\rm SL}(2, \RR)/\{\pm 1\}$ on the lightray with the action $z \mapsto g(z) = \frac{\alpha z+ \beta}{\bar \beta z + \bar \alpha}$ of $\GU/\{\pm 1\}$ on the circle, where $g$ now corresponds to the matrix 
$
 \left(
 \begin{matrix}
 \alpha & \beta\\
 \bar \beta & \bar \alpha
 \end{matrix}
 \right) \in \rm{SU}(1,1)
 $
 under the standard isomorphism between the groups ${\rm SL}(2, \RR)$ and $\GU$. 
 
 The axioms for a conformal field theory, i.e. net of operator algebras, on the 
 circle are completely analogous to those for the lightray. In the circle picture, it is more standard and natural to use the generators of $\GU$ called $L_0, L_{\pm 1}$, where $L_0$ is the generator of rotations $z \to e^{it} z, t \in \RR$. The requirement a4) is equivalent to the requirement that $L_0$ has non-negative spectrum. From a net on the circle, we may via the Caley transform always get a net on the lightray such that $P$ has non-negative spectrum, but not necessarily vice versa since the point at infinity is missing from the lightray. For the rest of the paper, we will assume the axioms on the circle. In sec. \ref{sec:sect5}, we will also need: 

\begin{enumerate}
\item[a6)] (Finite trace)
 $
 \tr e^{-\beta L_0} < \infty
$
 for $\beta>0$. 
 \end{enumerate}
 
 The above axioms (including the trace condition just mentioned) have a number of well-known consequences which are of interest for this paper: 
 \begin{enumerate}
 \item For each interval $\gA(I)' = \gA(I')$ (Haag duality \cite{longo3}).
 \item For each interval, the linear subspace $\gA(I) |\Omega_0 \rangle$ is dense in $\H$ (Reeh-Schlieder theorem). As a consequence, the 
 vector $|\Omega_0 \rangle$ is cyclic and separating for each local algebra $\gA(I)$, and we can apply Tomita-Takesaki theory to the pair 
 $(\gA(I), |\Omega_0 \rangle)$.
 \item The modular operator $\Delta$ associated with an open arc $I=(a,b)$ acts geometrically in the sense that 
 \ben
 \Delta^{it} = U(g_t), \quad g_t(z) = \frac{a(b-z) e^{-2\pi t} + b(z-a)}{(b-z)e^{-2\pi t} + (z-a)}, 
 \een
 (Hislop-Longo-theorem \cite{longo3}).  
 \item Each algebra $\gA(I)$ has in its central decomposition only hyperfinite type III$_1$ factors \cite{fredenhagen_5,buchholz4}. 
 \item $\overline{\bigcup_I \gA(I)} = \gB(\H)$ (irreducibility \cite{buchholz3}).
 \end{enumerate}
 
 Most of these axioms and results have a more or less obvious counterpart for graded local, i.e. ``fermionic'', theories, see e.g. \cite{bischoff}.
 
\medskip
\noindent
{\bf Example 4:} (Virasoro-net) The Virasoro algebra is the Lie-algebra with generators  $\{ L_n, \kappa \}_{n \in \ZZ}$  obeying
\ben
[L_n,L_m] = (n-m) L_{n+m} + \frac{\kappa}{12} n(n^2-1) \delta_{n,-m} , \quad 
[L_n, \kappa] = 0 .
\een  
A positive energy representation on a Hilbert space $\H$ is a representation such that (i) $L^*_n=L_{-n}$ (unitarity), (ii) $L_0$ is diagonalizable with 
non-negative eigenvalues, and (iii) the central element is represented by $\kappa = c1$. 
From now, we assume a positive energy representation. We assume that $\H$ 
contains a vacuum vector $|\Omega_0\rangle$ which is annihilated by $L_{-1}, L_0, L_1$, 
($\mathfrak{sl}(2,\RR)$-invariance) and which is a highest weight vector (of weight 0), i.e. $L_n |\Omega_0\rangle = 0$ for all 
$n >0$. One has the bound \cite{longo2,buchholz,goodman_wallach,goodman_wallach1}
\ben
\label{poly}
\|(1+L_0)^k L_n \Psi \| \le \sqrt{c/2}(|n|+1)^{k+3/2} \|(1+L_0)^{k+1} \Psi \|
\een
for $|\Psi \rangle \in \cV \equiv \bigcap_{k \ge 0} \cD (L_0^k) \subset \H$ and any natural number $k$.

One next defines from the Virasoro algebra the stress tensor on the unit circle $\SS$, identified with points $z=e^{2 \pi i u}, u \in \RR$ in $\CC$. 
The stress tensor is an operator valued distribution on $\H$ defined in the sense of distributions by the series
\ben
T(z) = -\frac{1}{2\pi} \sum_{n =-\infty}^\infty L_n z^{-n-2}. 
\een
More precisely, for a test function $f \in C^\infty(\SS)$ on the circle, 
it follows from \eqref{poly} that the corresponding smeared field 
\ben
T(f) = \int_{\SS} T(z) f(z) \dd z :=-\frac{1}{2\pi} \sum_{n =-\infty}^\infty  \left(\int_{\SS} z^{-n-2} f(z) \dd z \right)L_n
\een
is an operator defined e.g. on the dense invariant domain $\cV = \bigcap_{k \ge 0} \cD (L_0^k) \subset \H$ (which can be shown to 
be a common core for the operators $T(f)$) and the assignment 
$f \mapsto T(f)|\psi \rangle$ is continuous in the topologies on $C^\infty(\SS)$ and $\H$ for any vector in this domain.
Letting $\Gamma$ be the anti-linear involution 
\begin{equation}
\Gamma f(z) = -z^2 \overline{f(z)},
\end{equation} 
the smeared stress tensor is a self-adjoint 
operator on $\cD(L_0)$ for $f$ obeying the reality condition $\Gamma f = f$, and  one has $T(f)^* = T(\Gamma f)$ in general. 
It can be shown that the operators $e^{iT(f)}$ for real $f$ form a unitary projective representation of the (covering of the) group of 
orientation preserving diffeomorphisms (whose generators are the vector fields $f(z) \dd/\dd z$) on the circle. The Virasoro net is then 
defined by 
\ben
\gA_{\rm Vir}(I) = \{ e^{iT(f)} \mid f \in C^\infty(I), \Gamma f = f \}''
\een
where the double prime means the v. Neumann closure. The generators of this algebra hence correspond to diffeomorphisms acting trivially outside the arc 
$I \subset \SS$.

\subsection{Pointlike fields}

The standard setup of CFT commonly used in the physics literature is based on the use of pointlike fields rather than 
nets of algebras of bounded operators. Here we will sketch the connection. 
In fact, the full mathematical details of this connection 
are not understood in general, although in many important classes of examples, see \cite{longo2}. 

On the circle, one typically postulates the existence of local fields having the ``mode expansions'' 
\ben\label{smeared}
\phi(z) =  \frac{1}{\sqrt{2\pi}}
\sum_{n \in \bZ} \phi_n z^{-n-h} . 
\een
$h > 0$ is called the conformal dimension of the field. The field is typically 
``energy bounded'' i.e. that the modes $\phi_n, n \in \bZ$ of the field are linear operators on $\H_0$, which satisfy: 

\begin{assumption}
\label{ass2}
The local fields have a mode expansion \eqref{smeared} such that:
\begin{enumerate}
\item an energy bound of the type 
$ \|(1+L_0)^k \phi_n \Psi \| \le C(1+|n|)^{k+h-\frac12} \| (1+L_0)^{k+h-1} \Psi \|$ for all $n \in {\mathbb N}_0, |\Psi \rangle \in \H_0$, and 
for some $k \ge 0$, satisfying 
\item 
the commutation relations 
$[L_m, \phi_n]=((h-1)m-n)\phi_{n+m}$ for $|m| \le 1$ where $L_{-1},L_0,L_1$ are the generators of the 
action of $\GU$ on $\H_0$ and where $h \in \RR$ is called the conformal spin, satisfying 
\item if $| \Omega_0 \rangle \in \H_0$ is the vacuum vector, then $\phi_n | \Omega_0 \rangle = 0$ for $n>-h$, and satisfying
\item $\phi_n^*=\phi_{-n}$ for a self-adjoint local field.
\item The fields $\phi(f), {\rm supp}(f) \subset I$ should be affiliated with $\gA(I)$, i.e. there exists a sequence $B_n$ 
such that $\lim_n B_n |\Psi\rangle = \phi(f) |\Psi\rangle$ for all $\Psi \in \cV=\cap_k \cD(L_0^k)$.
\end{enumerate}
\end{assumption}

These properties imply that the smeared fields are operator valued tempered distributions on 
the domain $\cV=\cap_k \cD(L_0^k)$: Let $|\Psi \rangle \in \cV$. Then 1) gives, 
with $\phi(f) := \int_{\SS} \phi(z) f(z) \dd z$, 
\ben
\| (1+L_0)^k \phi(f) \Psi \| \le  C \| (1+L_0)^{k+h-1} \Psi \| \sum_{n \in \bZ}  | \widehat f_{-n-h} | \ (1+|n|)^{h-\frac12}  \le C_{\Psi} \sup_{j \le h} \| f^{(j)} \|_{L^\infty}, 
\een
because $|  \widehat f_{-n-h} |$ goes to zero for $|n| \to \infty$ faster than any inverse power. 
Thus, $\phi(f)$ is an operator valued distribution on the dense invariant domain $\cV$, 
which is in fact a common core for the operators $\phi(f)$. By the same type of estimate the properties 1) and 2) imply furthermore that 
$\phi(z)| \Omega_0 \rangle$ can be analytically continued to a $\H$-valued holomorphic function on $\bD^+$ 
with vector valued distributional boundary value on $\bS$.
It follows from the commutation relations 2) that $\H$ carries a  
strongly continuous unitary representation $U$ of $\GU$ generated by $L_0, L_{\pm 1}$, and this representation satisfies
transformation law
\ben\label{trafo}
i[L_m,\phi(z)] = z^{m+1} \frac{\dd}{\dd z}\phi(z) + h(\rho) z^m \phi(z) \quad \text{or} \quad
U(g) \phi(z) U(g)^* = [g'(z)]^h \phi(g(z)),
\een
where $g \in \widetilde{\GU}$ is in the covering group of the M\" obius group, and $g(z)$ its action on points $z$ of the circle.
For integer $h \in \bN_0$, we get a representation of $\GU/\{\pm 1\}$. The restriction of $U$ to the invariant subspace 
${\rm span}\{\phi_n |\Omega_0 \rangle = 0 \mid n\le -h\}$ is a discrete series representation (see e.g. IX, para. 3 of \cite{Lang}). 
It also follows that primary fields can and will be normalized so that
\ben\label{ddt}
\langle \Omega_0 | \phi(x) \phi(y) | \Omega_0 \rangle = \frac{ e^{-i\pi h}}{2\pi \ (x-y)^{2h}}. 
\een
In particular, we see that the field can be local only if the dimension $h$ is a natural number. 
Fermionic fields are not local but satisfy a graded locality. In that case $h \in \half {\mathbb N}_0$. 

\medskip
\noindent
{\bf Example 5:} (Stress tensor)
The stress tensor $T(z)$ affiliated with the Virasoro net of central charge $c>0$ is a pointlike field of dimension $h=2$ satisfying the above assumptions.

\medskip
\noindent
{\bf Example 6:} ($U(1)$-current, see e.g. \cite{buchholz, bischoff})
The net of the free $U(1)$ current on the circle can be defined e.g. starting from the Lie-algebra generated by a central element $1$ and the ``modes'' 
$j_n, n \in \ZZ$ defined by $[j_n, j_m]=in \delta_{n,-m} 1$ with *-operation $j_n^*=j_{-n}$. The Hilbert space $\H$ is the closure of the linear span 
of $j_{n_1} \dots j_{n_k}|\Omega_0\rangle, n_1\le \dots \le n_k \le -1$ on which the action of the $J_m$'s is obtained via the commutation relations and 
the condition $j_n |\Omega_0\rangle = 0$ for $n>-1$. One sets $L_n = \frac12 \sum_{m \in \bZ} :j_{n-m}j_m:$, 
where here and in the following, the normal ordering sign $: \ , \ :$ means that 
modes with index $m>-1$ (or $>h$ if the field has dimension $h$) are always put to the right of the modes with index $n-m\le -1$. The $L_n$'s the satisfy a Virasoro algebra of central charge $c=1$.

It can be checked that the corresponding current 
\ben
j(z) = \frac{1}{\sqrt{2\pi}} \sum_{n \in \bZ}  j_n z^{-n-1} 
\een
satisfies the above assumptions with $h=1$ and
is hence an operator valued distribution satisfying $j(z)^* = z^2 J(z)$ and $[j(z), j(w)]=i \delta'(z-w)$. For any test-function $f$, 
the smeared operator $j(f) := \int_{\SS} j(z) f(z) \dd z$ has a dense set of analytic vectors (a space of such vectors is spanned by the eigenvectors of $L_0$), 
and hence is essentially self-adjoint by Nelson's analytic vector theorem. Hence, we can unambiguously define the Weyl operators 
\ben
W(f) = e^{ij(f)}, \quad \text{$f \in C^\infty_\Gamma(\SS)$,}
\een
(here $\Gamma f(z) = -\overline{f(z)}$ and $C^\infty_\Gamma(\SS)$ is the set of invariant elements under $\Gamma$),
satisfying the Weyl relations 
\ben
W(f) W(g) = e^{iC(f,g)/2} W(f+g), \quad W(f)^* = W(-f), \quad C(f,g) =\half \int_\SS (g' f-f'g) \dd z . 
\een
The corresponding net of v. Neumann algebras is defined by 
\ben
\A_{U(1)}(I) \equiv \{ W(f) \mid f \in C^\infty_\Gamma(I) \}''
\een
where $I \subset \SS$ is an open arc of the circle or a union thereof, 
and the double prime means the weak closure (double commutant). The local, unbounded field operators $j(f), \supp f \subset I$ are not contained in- but are affiliated with these algebras. 

\medskip
\noindent
{\bf Example 7:} (Free Fermi net, see e.g. \cite{Rehren:2012wa,bischoff})
This net is constructed 
starting from the Clifford algebra generated by a central element $1$ and the ``modes'' 
$\psi_n, n \in \ZZ+\frac12$ subject to the relations $\psi_n \psi_m + \psi_m \psi_n =\delta_{n,-m} 1, \psi_n^*=\psi_{-n}$. The vacuum Hilbert space $\H_{\rm NS}$ is the closure of the linear span of the vectors $\psi_{n_1} \cdots \psi_{n_k}|\Omega_{\rm NS} \rangle, n_1 < n_2 < \dots < 0$. A $*$-representation is defined setting $\psi_n |\Omega_{\rm NS} \rangle = 0$ for all $n \ge 0$ using the relations to define the action of an arbitrary $\psi_n$. The state $|\Omega_{\rm NS}\rangle$ is in this context called the ``Neveu-Schwarz-vacuum''. The operators $L_n$ are defined by $L_n= \sum_{m \in \bZ + \frac12} m :\psi_{-m+n} \psi_m:$ which generate an action 
of the Virasoro algebra (in particular of the Lie algebra of $\GU$ generated by $L_n, n = -1,0,1$), at central charge $c=\frac{1}{2}$.  
The corresponding field 
\ben
\label{NS}
\psi(z) = \frac{1}{\sqrt{2\pi}} \sum_{n \in \bZ}  \psi_{n-\frac12}  z^{-n} 
\een
is hence an operator valued distribution. It satisfies $\psi(z)^* = z\psi (z)$ and $\psi(z)^*\psi(w) + \psi(w) \psi(z)^*=\delta(z-w)1$. For any test-function $f$, 
the smeared operator $\psi(f):= \int_{\SS} \psi(z) f(z) \dd z$ is in fact a bounded operator satisfying the canonical anti-commutation relations 
\ben
\psi(f) \psi(g) + \psi(g) \psi(f) = -(\Gamma f, g) 1, \quad \psi(f)^* = \psi(\Gamma f), \quad \Gamma(f) = -z^{-1} \overline{f(z)} . 
\een 
The corresponding net of v. Neumann algebras is defined as the CAR-algebra \cite{araki_5}
\ben\label{CAR}
\gA_{\rm Fermi}(I) \equiv \{ \psi(f)  \mid f \in C^\infty_0(I) \}''. 
\een
The net of local observables is not a local net but a graded local net, see e.g. \cite{bischoff}.
There is another representation of the same net $\gA_{\rm Fermi}$, called the ``Ramond'' representation. It is given by the 
integer moded expansion
\ben
\label{R}
\psi(z) = \frac{1}{\sqrt{2\pi}} \sum_{n \in \bZ}  \psi_n z^{-n-\half} , 
\een
where the modes satisfy the same relations as before. The Hilbert space $\H_{\rm R}$ is constructed as the linear span of the vectors $\psi_{n_1} \cdots \psi_{n_k}|\Omega_{\rm R} \rangle, n_1 < n_2 < \dots \le 0$ setting $\psi_n |\Omega_{\rm R} \rangle = 0$ for all $n > 0$ using the relations to define the action of an arbitrary $\psi_n$. 
The Virasoro generators in the Ramond representation $\H_{\rm R}$  are $L_n
= \sum_{m \in \bZ} m :\psi_{-m+n} \psi_m:$.

\section{Modular operators for conformal nets on $\bS$}\label{sect2}

In this section we give a first prescription for computing modular operators of  chiral conformal nets on $\bS$ satisfying some natural extra conditions. It is related naturally 
to the matrix elements $\langle \Omega_0 | \phi(x) \Delta^{it} \phi(y) | \Omega_0 \rangle$, but leaves in general certain ambiguities that preclude so far their explicit calculation. This  difficulty can be overcome to a certain extent by our second method, presented in 
sec. \ref{sec:sect5}, more directly related to the matrix element $\langle \Omega | \phi(x) (\log \Delta) \phi(y) | \Omega \rangle$. Since the material here will form the basis of our discussion in sec. \ref{sec:sect5}, and since the arguments are also of independent interest, we 
nevertheless present this approach first. 

\subsection{General results }
\label{localfields}
Quite generally, if $\gM$ is a v. Neumann algebra in standard form with cyclic and separating vector $\Omega$, then if $X,Y \in \gM$, the Fourier transform
\ben\label{FdefFourier}
        \Gamma_{X,Y}(s) \equiv \int_\RR \langle \Omega | X \Delta^{it} Y \Omega \rangle \, e^{its} \frac{\dd t}{2 \pi}
\een
is well-defined in the sense of a tempered distribution in the variable $s$ -- in fact for $Y=X^*$, $\Gamma_{X,X^*}(s) \dd s$ is a positive Radon measure on $\RR$, see sec. 5.3 of \cite{Bratteli}. 
By the Fourier inversion theorem, the operator
$\Delta^{it}$ is hence fully characterized provided we know $\Gamma_{X,Y}(s)$ for all $X,Y \in \gM$ and all $s$, i.e. as a distribution in $s$. 

Using the KMS condition \eqref{KMS} after shifting the integration contour from the real axis $\RR$ to the line $\RR-i$ parallel to the real axis 
immediately gives
\ben
\label{KMS1}
\Gamma_{X,Y}(s) = e^s \Gamma_{Y,X}(-s). 
\een
On the other hand, for $X \in \gM, Y \in \gM'$ or vice versa,
we get
\ben
\label{KMS6}
\Gamma_{X,Y}(s) = \Gamma_{Y,X}(-s)
\een
using that the modular flow $\sigma^t(X) = \Delta^{it} X \Delta^{-it}$ preserves $\gM, \gM'$.

We now want to describe how the extra structure of chiral conformal field theory can help to characterize $\Gamma_{X,Y}(s)$. 
The case we want to consider is the v. Neumann algebra $\gM=\A(A)$,  associated with a region $A$ consisting of $p$ open arcs.
The vector under consideration is the vacuum, $|\Omega\rangle = | \Omega_0 \rangle$.
It seems that the information is
most easily retrieved if instead of bounded operators $X,Y$, we work with point-like unbounded field operators as described in the 
previous section.  We define for a generic primary field $\phi$:
\ben\label{Gdef0}
\Gamma(s; x, y) \equiv \int_\RR \langle \Omega_0 | \phi(x) \Delta^{it} \phi(y) | \Omega_0 \rangle \, e^{its} \frac{\dd t}{2 \pi} .
\een
Here, $x,y \in \bS$ are to be smeared with test functions on $A$ or $A'$, and $\Delta \equiv \Delta_A$ is the modular operator 
in the vacuum state for the multi-interval/arc $A$. The quantity $\Gamma$ should be considered as analogous to \eqref{FdefFourier}. 

\medskip
\noindent
{\bf Example 8:} For one arc, $A=(a,b)$, the modular flow of a local primary field of dimension $h$ is given by the Hislop-Longo theorem \cite{hislop,longo3} as
\ben\label{longo}
\Delta^{it} \phi(x) \Delta^{-it} = [g_t'(x)]^h \phi(g_t(x)) , \quad g_t(x) = \frac{a(b-x) e^{-2\pi t} + b(x-a)}{(b-x)e^{-2\pi t} + (x-a)}.
\een
Therefore $\Gamma$ can be found from \eqref{Gdef0} and \eqref{ddt}, giving for $x,y \in (a,b)$
\ben
\label{1I}
\begin{split}
\Gamma(s,x,y) = & \frac{e^{-is/2-i\pi h}| \Gamma ( h-\tfrac{is}{2\pi})|^2}{(2\pi)^2\Gamma(2h)} \left( \frac{(b-a)^2}{(x-a)(x-b)(y-a)(y-b)} \right)^{h}  \cdot \\
&\cdot \exp \left(\frac{is}{2\pi} \log \frac{(x-a)(y-b)}{(x-b)(y-a)}\right).
\end{split}
\een

The example suggests that the behavior of $\Gamma(s,x,y)$ near a boundary point $q_i$ of a multi-interval could be $(x-q_i)^{-h}$. This is supported by the following lemma, formulated in the circle picture.

\begin{lemma}
\label{lemma1}
Under the assumptions on the CFT given in the previous subsections:
\begin{enumerate}
\item
If $f \in C^\infty_0(\RR)$ and $\Gamma(f; x, y) \equiv \int \Gamma(s; x, y) f(s)  \dd s$, then $G(f; x, y)$ is smooth in $x,y \in \bS$ away from the $2p$ end-points  of the $p$ intervals $I_j$. Moreover, for any $0<\epsilon\le \frac12$,
\ben\label{firstbound}
\begin{split}
&\left| \Gamma(f,x,y) \prod_{n=1}^p (x-a_n)^{h} (y-a_n)^{h} (x-b_n)^{h} (y-b_n)^{h} \right| \\
& \le C \epsilon^{-2h} \, \sup_s \left(  e^{-(\frac12-\epsilon)|s|} e^{\sigma s/2} |f(s)| \right)  
\end{split}
\een
for some constant $C$ only depending on the end-points. Here, $\sigma = +1$ if both $x,y$ are in $A$, $\sigma=-1$ if $x,y$ are in $A'=\bS \setminus \bar A$.
\item
We have the KMS condition
\ben\label{KMS2}
\Gamma(s; x, y) =e^s \Gamma(-s; y, x)
\een
in the sense of distributions for $x,y \in A$. When $x \in A, y \in A'$ (or vice versa), we have instead 
\ben\label{KMS3}
\Gamma(s; x, y) = \Gamma(-s; y, x).
\een
\end{enumerate}
\end{lemma}

\begin{proof}
1) We will compare the quantity $\Gamma$ of an arbitrary multi-arc $A$ to that corresponding to a single arc. First we assume $x,y \in A$. Let $I$ be the largest
arc contained in $A$ that is symmetric around $x$, and $J$ the largest interval contained in $A$ symmetric around $y$. If, for example, 
$c_j$ resp. $c_k$ are to the left of $x$ resp. $y$, then $I=(c_j, c_j^{-1}x^2)$ resp. $J=(c_k, c_k^{-1}y^2)$.  
We let $\Delta_A, \Delta_J, \Delta_I$ be the modular operators for the corresponding local algebras. Then we use the well-known 
operator inequality $\Delta_I^{\alpha} \ge \Delta_A^{\alpha}$ for $0 \le \alpha \le 1$. 
This follows from the fact that $\A(A) \supset \A(I)$, which is exploited as follows. Quite generally, let $\gM_i$ be two v. Neumann algebras on 
the same Hilbert space with common cyclic and separating vector $|\Omega \rangle$. 
   We let $S_i$ be the Tomita operators for $\gM_i$ with polar decompositions $S_i = J_i \Delta_i^{1/2}$.  Note that, 
if $\gM_2 \subset \gM_1$, then ${\mathcal D}(S_2) \subset {\mathcal D}(S_1)$. 
The set ${\mathcal D}(S_1)$ is a Hilbert space called $\H_1$ with respect to the inner product (graph norm)
\ben
\label{H1}
( \Phi , \Psi ) = \langle \Phi | \Psi \rangle + e^u \langle S_1 \Psi | S_1 \Phi \rangle = \langle \Phi | (1+e^u \Delta_1) \Psi \rangle . 
\een
Letting $I:\H_1 \to {\mathcal D}(S_1)$ 
be the identification map, one shows that  $I^{-1} {\mathcal D}(S_2)$ is a closed subspace $\H_2 \subset \H_1$ with associated orthogonal projection $P_2$. 
The operators $V_j= I^{-1}(1+ e^u\Delta_j)^{-1/2}$ are isometries from $\H$ to $\H_j$ ($j=1,2$) and their adjoints are $V_j^*= (1+ e^u\Delta_j)^{1/2} I P_j$
(with $P_1=1$). There follow the relations
\bena
I P_j I^* &=& I V_j^{} V_j^* I^* = (1 + e^u\Delta_j)^{-1}, \quad j=1,2\\
I^* &=& I^{-1}(1+ e^u\Delta_1)^{-1}, 
\eena
which can already be found in \cite{fredenhagen_5}. 

We multiply the first relation from the right with $X  \in \gM_2$ and from the left with $Y^* \in \gM_2$ and take the expectation value in 
the state $|\Omega\rangle$. Then we use the second equation and obtain
\ben\label{frede}
\begin{split}
&\langle \Omega |Y^* (1 + e^u\Delta_1)^{-1} X  | \Omega \rangle -  \langle \Omega | Y^* (1 + e^u\Delta_2)^{-1} X  | \Omega \rangle\\
&=  \langle (1- P_2) I^{-1} (1 + e^u\Delta_1)^{-1} Y | (1- P_2) I^{-1} (1 + e^u\Delta_1)^{-1} X \Omega \rangle \ .
\end{split}
\een
The fact that the right side is manifestly non-negative for $X=Y$ implies $(1 +e^u \Delta_2)^{-1} \le (1 +e^u \Delta_1)^{-1}$, and that, combined with 
the operator identity 
\ben
\Delta^\alpha_1 - \Delta_2^\alpha = \frac{\sin \pi\alpha}{\pi} \int_{-\infty}^\infty  e^{-\alpha u} [(1+ e^u\Delta_2)^{-1} - (1 + e^u\Delta_1)^{-1} ] \dd u \le 0 \ 
\een
for $0<\alpha<1$ gives the claim. Therefore $\Delta^{-\alpha/2}_I  \Delta^{\alpha/2}_A(\Delta^{-\alpha/2}_I  \Delta^{\alpha/2}_A)^*
=\Delta^{-\alpha/2}_I  \Delta^{\alpha}_A \Delta^{-\alpha/2}_I \le 1$
implying $\|\Delta_A^{\alpha/2} \Delta_I^{-\alpha/2}\| \le 1$, and similarly for $J$. By the functional calculus, if $\dd E(\lambda)$ is the spectral resolution of 
$\log \Delta_A$:
\ben
\begin{split}
\Gamma(f,x,y) &=(2\pi)^{-1} \int_\RR e^{its} f(s) \int_\RR e^{i\lambda t}  \langle \Omega_0 | \phi(x) \dd E(\lambda) \phi(y) | \Omega_0\rangle \ \dd \lambda \dd s \\
&=\langle \Omega_0 | \phi(x) f(-\log \Delta_A) \phi(y) | \Omega_0\rangle ,
\end{split}
\een
and then by the Cauchy-Schwarz inequality 
\ben
\label{CS}
\begin{split}
|\Gamma(f,x,y)| 
&= \bigg|\langle \Omega_0 | \phi(x) f(-\log \Delta_A) \phi(y) | \Omega_0\rangle \bigg| \\
&\le \| \Delta_A^{-\alpha} f(-\log \Delta_A^{}) \| \,  
\| \Delta_A^{\frac{\alpha}{2}} \Delta_I^{-\frac{\alpha}{2}} \| \, 
\| \Delta_A^{\frac{\alpha}{2}} \Delta_J^{-\frac{\alpha}{2}} \| \,
\| \Delta_I^{\frac{\alpha}{2}} \phi(x) \Omega_0 \| \, \| \Delta_J^{\frac{\alpha}{2}} \phi(y) \Omega_0 \| \\
&\le \left( \sup_s e^{\alpha s} |f(s)| \right) \| \Delta_I^{\frac{\alpha}{2}} \phi(x) \Omega_0 \| \, \| \Delta_J^{\frac{\alpha}{2}} \phi(y) \Omega_0 \| \ . 
\end{split}
\een
Since $\phi(x) |\Omega_0 \rangle$ can be analytically continued to a $\H_0$-valued 
holomorphic function inside the unit disk $\bD^+$ (by the mode expansion of $\phi$, see the previous subsection), 
the Hislop-Longo theorem applied to the modular flow $\sigma_I^t$ of $I$ can similarly be 
continued to imaginary flow time parameter $t=-i\alpha$ (see eq. \eqref{longo}).
In combination with \eqref{ddt} we thereby obtain
\ben
\label{L}
\begin{split}
\| \Delta_I^{\frac{\alpha}{2}} \phi(x) \Omega_0 \|^2 
& = \langle \Omega_0 | \phi(x)^* \Delta_I^\alpha \phi(x) |\Omega_0 \rangle \\
&= \langle \Omega_0 | \phi(x)^* \sigma_I^{-i\alpha} (\phi(x)) |\Omega_0 \rangle \\
&= x^{2h} \langle \Omega_0 | \phi(x) \sigma_I^{-i\alpha} (\phi(x)) |\Omega_0 \rangle \\
&= (2\pi)^{-1}  |\dd g^I_{-i\alpha}(x)/\dd x \, (x - g_{-i\alpha}^I(x))^{-2}|^{h} \\
&\lesssim (\sin \pi \alpha)^{-2h} |x-c_j|^{-2h},
 \end{split}
\een
 and similarly for $y$.  Since this holds for any pair of end points, 
\ben\label{firstbound1}
\begin{split}
&\left| \Gamma(f,x,y) \prod_{n=1}^p (x-a_n)^{h} (y-a_n)^{h} (x-b_n)^{h} (y-b_n)^{h} \right| \\
&\lesssim \, (\sin \pi \alpha)^{-2h}  \sup_s \left(  e^{\alpha s} |f(s)| \right)  
\end{split}
\een
Now we split the testfunction $f=f_++f_-$, where the testfunction $f_-$ has support in $(-\infty,c]$, and 
the testfunction $f_+$ has support in $[-c,\infty)$. Then, for the contribution from $f_-$, we choose 
$\alpha=1-\epsilon$, whereas for the contribution from $f_+$, we choose $\alpha=\epsilon$. As a consequence, we 
find that \eqref{firstbound} holds.

This shows in particular boundedness in $x,y$ of $\Gamma(f,x,y)$ away from the endpoints. A similar estimation can be made for descendant fields, i.e. 
the derivatives of the fields $\phi(x), \phi(y)$, and this shows smoothness. 
This finishes the proof of 1) when $x,y \in A$.

 To cover the other case, we note that the modular operator for $\gA(A)' \supset \gA(A')$ is related to that 
of $\gM_1$ by $\Delta_{1}' = \Delta_1^{-1}$. Then the other case of 1) follow by the same 
argument, namely, $x,y \in A'$ we take $I'$ be the largest
interval contained in $A'$ that is symmetric around $x$, and $J'$ the largest interval contained in $A'$ symmetric around $y$ and proceed in the same way as before.

2) Eq. \eqref{KMS2} follows from the KMS condition \eqref{KMS} for bounded operators $X,Y \in \pi_0(\gA(A))''$ because we are assuming about the local fields 
that for test-functions $f,g$ supported in $A$ there exist $X_n,Y_n \in \pi_0(\gA(A))''$ with the property that $\lim_n X_n |\Omega_0\rangle = \phi(f)|\Omega_0\rangle$ and $\lim_n Y_n |\Omega_0\rangle = \phi(g)|\Omega_0\rangle$ in the strong topology. 
A similar remark applies to \eqref{KMS3}.
\end{proof}

We next make analytic continuations in the variables $x,y$. This gives us the following:
For fixed $x \in A$, and fixed test-function $f(s)$, 
the function $y \mapsto \Gamma(f; x, y)$ has an analytic extension to a holomorphic function of $y$ inside the unit disk $\bD^+ = \{ z \in \CC \mid |z|<1 \}$, by 
the mode expansion of $\phi$. 
Similarly, for fixed $y \in A$,
the function $x \mapsto \Gamma(s; x, y)$ has an analytic extension to $x$ outside the unit disk $\bD^- = \{ z \in \CC \mid |z|>1 \}$.
How to extend to $y \in \bD^-$ or $x \in \bD^+$? The idea is that $\Gamma(f, y, x)$
is analytic in this domain, so we try to paste $\Gamma(s; x, y)$ and $\Gamma(f, y, x)$ together across the boundary of the disk and hope that 
we get an analytic function that way. This will turn out to be the case on account of the KMS condition which relates the two quantities in precisely the right way. 

To this end, we define the following auxiliary quantity for {\em fixed} $x \in A, s \in \RR$:
\ben
\label{cases1}
K(y) \equiv
\begin{cases}
\Gamma(+s,x, y) & \text{if $y \in \bD^+$,}\\
\Gamma(-s,y,x) & \text{if $y \in \bD^-$.}
\end{cases}
\een
We similarly define, for {\em fixed} $y \in A, s \in \RR$:
\ben
\label{cases2}
H(x) \equiv
\begin{cases}
\Gamma(+s,x, y) & \text{if $x \in \bD^-$,}\\
\Gamma(-s,y,x) & \text{if $x \in \bD^+$.}
\end{cases}
\een
Note that $H$ implicitly also depends on the choice of $y,s$ and $K$ on the choice of $x,s$, but 
we suppress this since we are, for the moment, only interested in the dependence of $H$ on $x$
and of $K$ on $y$. Furthermore, note that $H$ and $K$ are, a priori defined only as holomorphic functions on the union
$\bD^+ \cup \bD^- = \CC \setminus \bS$, i.e. the complex plane minus the circle. On the circle, we 
define the boundary values from the inside resp. the outside of the disk ($\pm$):
\ben
H^\pm(x) = \lim_{z \to x, z \in \bD^\pm} H(z), \quad 
K^\pm(y) = \lim_{z \to y, z \in \bD^\pm} K(z), \quad x,y \in \bS. 
\een
Then the ``jump conditions'' \eqref{KMS2}, \eqref{KMS3} imply that 
\ben\label{RH}
H^+(x) = 
\begin{cases}
H^-(x) & \text{if $x \in A'$,}\\
e^{-s} H^-(x) & \text{if $x \in A$,}
\end{cases}
\quad 
K^+(y) = 
\begin{cases}
K^-(y) & \text{if $y \in A'$,}\\
e^{+s} K^-(x) & \text{if $y \in A$.}
\end{cases}
\een
Thus, both $H$ resp. $K$ are solutions to a Riemann-Hilbert-problem (across the contour $\bS$). These problems are essentially completely understood, 
see e.g. \cite{mush,gakh}. The number and type of solutions depends in general on the specification of the behavior of $H$ resp. $K$ 
near the boundary points $\{ q_j \}$ of the multi-arc $A = \cup_j (a_j,b_j)$ and at infinity, see e.g. \cite{mush} (para. 79, pp 230) or 
\cite{gakh} (para. 42, pp 420). In the case at hand, this behavior is restricted by \eqref{firstbound} and by the mode expansions of the fields. 

A minor technical complication arises at this stage due to the fact that, as functions of $s$, both $H$ resp. $K$ are only defined in the distributional sense, 
and our bound \eqref{firstbound} likewise also involves a test-function in $s$. This complication would make the direct application of the 
results in \cite{gakh,mush} somewhat cumbersome, so we give an explicit analysis of the implications imposed by the Riemann-Hilbert problem in the case at 
hand taking into account this complication.

First we define the shorthands
\ben\label{Pidef}
\Pi_a(x) = \prod_{j=1}^p (x-a_j), 
\quad \Pi_b(x) = \prod_{j=1}^p (x-b_j), 
\een
and 
\ben
\label{Zdef}
Z_\pm(x) = \frac{1}{2\pi} \log \left( \pm
\prod_{j=1}^p \frac{(x-a_j)}{(x-b_j)}
\right).
\een
Notice that $Z_+(x(1 \mp \epsilon)) = Z_-(x(1 \mp \epsilon)) \mp i/2$ when $\epsilon \to 0^+$, and that $Z_+$ has branch cuts on $A$, while 
$Z_-$ has branch cuts on $A' = \bS \setminus \overline A$.  
Then we define 
\ben
\label{cases4}
\begin{split}
\widetilde K(y) &\equiv (\Pi_a(y) \Pi_b(y))^{h} \int f(s) 
e^{is(Z_+(x)-Z_+(y))}
K(y,s) \ \dd s, \\
\widetilde H(x) &\equiv (\Pi_a(x) \Pi_b(x))^{h} \int f(s)
e^{is(Z_+(x)-Z_+(y))}
H(x,s) \ \dd s.
\end{split}
\een
We have: 

\begin{lemma}\label{lem7}
$\widetilde K(y)$ is a polynomial in $y$ of degree at most $2(p-1)h$, 
$\widetilde H(x)$ is a polynomial in $x$ of degree at most $2(p-1)h$.
\end{lemma}

\begin{proof}
We first consider $\widetilde K(y)$. The function $\log \frac{\Pi_a(y)}{\Pi_b(y)} = 2\pi Z_+(y)$ jumps by $+2\pi i$ as $y$ crosses $A$ from the inside of 
the unit disk to the outside. That jump compensates precisely the jump \eqref{KMS2} implied by the KMS condition, so that $\widetilde K(y)$ is continuous across $A$. Similarly, \eqref{KMS3} implies that $\widetilde K(y)$ is continuous across the complement $A'$ as a function of $y$, too. Therefore, 
by the edge-of-the-wedge theorem, $\widetilde K$ is an analytic function of $y$ in the entire complex plane minus the boundary points of the intervals. 
Since $\widetilde K$ is a tempered distribution on $\bS$ which away from the boundary points is the boundary value of an analytic function (from either inside or outside the unit disk), it follows that $\widetilde K$ cannot have any essential singularities at the boundary points $y \in \{q_j\}$. Indeed, at any given boundary point, say $a_j$, since we have a tempered distribution, there exists a natural number $N$, such that, if we multiply $\widetilde K(y)$ by $(y-a_j)^N$, we get a 
continuous function near $a_j$ in $y\in \bS$. By the edge of the wedge theorem, this function, being a boundary value from both inside and outside the disk, must 
be holomorphic near $a_j$. Actually, by \eqref{firstbound} we know the factor $(\Pi_a(y) \Pi_b(y))^{h}$ cancels the potential blow up near $a_j$
when $y \in A$, and therefore we can actually choose $N=0$. Since this argument can be repeated for any other boundary point, we learn that  the function $\widetilde K(y)$ is analytic in $y$ throughout the entire complex plane. A similar statement holds for $x$ and $y$ interchanged and with $\widetilde K$ and $\widetilde H$ interchanged.

We now establish a bound on the modulus of $\widetilde K(y)$ for $|y| \to \infty$. We learn 
from the mode expansions and properties of the fields that $y^{-2ph+2h} \widetilde K(y)$
remains bounded. Thus, we conclude that $|\widetilde K(y)| \lesssim |y|^{2hp-2h}$
throughout the entire complex plane for some new constant possibly depending on $x \in A$ and on the test-function $f$. 
Therefore, $\widetilde K(y)$ must for fixed $x \in A$ be a polynomial in $y$ of degree at most $2ph-2h$. We may repeat the same argument with the roles of $x$ and $y$ and of $\widetilde K$ and $\widetilde H$ reversed, and this finishes the proof. 
\end{proof}

The next lemma is a straightforward consequence of the preceding two lemmas. 

\begin{lemma}
\label{lem4}
As a distribution on $(x,y) \in \bS \times \bS$ (and $s \in \RR$), we have (for $Z_-$, see \eqref{Zdef}):
\ben\label{Gform}
\Gamma(s,x,y) =(\Pi_a(x) \Pi_b(x) \Pi_a(y) \Pi_b(y))^{-h}  \sum_{m,n=0}^{2(p-1)h} c_{mn}^{}(s) q_m(x) \overline{q_n(y)} e^{is[Z_-(y)-Z_-(x)]} ,
\een
where $\epsilon^{2h} e^{(\frac12-\epsilon)|s|-s/2} c_{mn}(s) \in L^1(\RR, \dd s)$ for each $\frac12 \ge \epsilon>0$, with 
uniformly bounded $L^1$-norm in $\epsilon$, and where $q_n$ are polynomials of degree $n$. 
\end{lemma}

\begin{proof}
We consider the distributional boundary values for $x \to A$ from within $\bD^-$
and for $y \to A$ from within $\bD^+$, respectively in the following expression
\ben
\widetilde \Gamma(s,x,y) = \lim_{x,y \to A} (\Pi_a(x) \Pi_b(x) \Pi_a(y) \Pi_b(y))^{h} e^{is[Z_+(x)-Z_+(y)]}
\Gamma(s,x,y). 
\een
This boundary value prescription coincides with that for $\Gamma(s,x,y)$ and thus the right side is well-defined as 
a distribution in on $A \times A$ (after smearing in $s$ against a testfunction $f(s)$), by elementary results on 
products of distributions that are boundary values of analytic functions. 
By lemma \ref{lem7}, $\widetilde \Gamma(f,x,y)$ is a polynomial both in $x$ and in $y$. 
The inequality \eqref{firstbound} and the definition of $Z_+$ \eqref{Zdef} gives the upper bound 
\ben
|\widetilde \Gamma(f,x,y)| \le C \epsilon^{-2h} \sup_s \left(  e^{-(\frac12-\epsilon)|s|-\frac12 s} |f(s)| \right)
\een
on this polynomial for all $x,y \in A$. Since the coefficients, $a_{mn}(f)$, of this polynomial can be reconstructed 
by interpolation from the values $\widetilde \Gamma(f, x_\alpha, y_\beta)$ for $2h(p-1)$ interpolation points $x_\alpha$ and 
$2h(p-1)$ interpolation points $y_\beta$ from $A$, this upper bound also holds for $a_{mn}(f)$. 
By the well-known duality between 
the Banach spaces $L^1(\RR)$ and $L^\infty(\RR)$ we can interpret this as saying that 
$a_{mn}(s) e^{(\frac12-\epsilon)|s|+\frac12 s}$ is a function in $L^1(\RR)$ with norm bounded from above by $C \epsilon^{-2h}$. 
Now set $c_{nm}(s) = a_{nm}(s) e^s$ and use the relationship $Z_+(x(1 \mp \epsilon)) = Z_-(x(1 \mp \epsilon)) \mp i/2$ when $\epsilon \to 0^+$. 
Then the proposition follows after expressing $\Gamma$ in terms of $\widetilde \Gamma$.
\end{proof}

If we wish, we can at this stage take an inverse Fourier transform of $G$ in 
$s$ to get a general expression for $\langle \Omega_0 |\phi(x) \Delta^{it} \phi(y) | \Omega_0 \rangle$. We set
\ben
\widehat{c}_{mn}(t)= \int_\RR e^{-its} c_{mn}(s) \dd s, 
\een
and we may take the polynomials $q$ in \eqref{Gform} as monomials, again for simplicity of notation. Then we immediately get:

\begin{proposition}
\label{lem6}
1) As a distribution in $(x,y) \in A \times A$ (with boundary value prescription $(x,y) \in \bD^- \times \bD^+ \to \bS \times \bS$ understood)
\ben\label{eq:Gform}
\begin{split}
\langle \Omega_0 |\phi(x) \Delta^{it} \phi(y) | \Omega_0 \rangle =& \left( \Pi_a(x) \Pi_b(x) \Pi_a(y) \Pi_b(y)  \right)^{-h} \cdot \\
& \cdot \sum_{m,n=0}^{2(p-1)h}  \widehat{c}_{mn}^{} \left(t +  Z(x) - Z(y)  \right) x^m y^n,
\end{split}
\een
where $Z \equiv Z_-$ is defined in \eqref{Zdef}, where $\widehat{c}_{mn}(t)$ is analytic in the strip $\{ t \in \CC \mid -1< \Im(t) < 0 \}$. There, it satisfies a bound
\ben
\label{sinbound}
|\widehat{c}_{mn}(t)| \lesssim [\sin (\pi \Im t)]^{-2h}
\een
and for real $t$ satisfies the property (in the distributional sense as a boundary value)
\ben
\label{KMS4}
\widehat{c}_{mn}(t-i) = \widehat{c}_{nm}(-t) = \overline{\widehat{c}_{mn}(t)}. 
\een
2) We must have:
\ben
\label{constraint}
\frac{(-1)^h}{2\pi} \left( \frac{Q(x,y)}{2\sinh \pi (Z(x)-Z(y)-i0)} \right)^{2h} =  \sum_{m,n=0}^{2(p-1)h}  \widehat{c}_{mn}^{} \left(Z(x)-Z(y)-i0 \right) x^m y^n 
\een
in the distributional sense (for $x,y \in A \subset \bS$), where the bi-variate polynomial $Q$ is as in \eqref{Qdef}. 
\end{proposition}
\begin{proof}
1) The formula \eqref{eq:Gform} follows directly from lemma \ref{lem4}. In particular, the claimed analyticity and bound \eqref{sinbound} follow from the corresponding bounds on $c_{mn}(s)$. The formula \eqref{KMS4} follows from the KMS-condition. 2) For $t=0$ we evidently have $\Delta^{it}=1$. This condition gives a non-trivial constraint on the functions $\widehat{c}_{nm}$. Introduce the quantity $Z$ as in \eqref{Zdef} and 
\ben
\label{Qdef}
Q(x,y) = \frac{\prod_{j=1}^p (x-a_j)(y-b_j) - \prod_{j=1}^p (y-a_j)(x-b_j)}{x-y}.
\een
Note that $Q(x,y)$ is a polynomial in $x,y$ of degree $2(p-1)$ in each variable.
We then get 2).
\end{proof}

\medskip
\noindent

\begin{remark}\label{rem6}
The domain of analyticity of $\widehat{c}_{mn}^{}$ is large enough to permit us to take the limit 
$|y| \to \infty$ or $|x| \to 0$. The constraint then confirms that $\widehat{c}_{mn}^{}(t)=0$ when $m,n > 2h(p-1)$.
\end{remark}

As we will see, in certain special cases eq. \eqref{constraint} and the properties given in proposition \ref{lem6}
suffice to determine $\widehat{c}_{mn}^{}$ uniquely. For instance we will see in subsec. \ref{NSvac} that for a free fermion, 
the information we have obtained uniquely fixes the modular flow. For the $U(1)$-current, the proposition is however already less restrictive, although
we are still able to get some results in subsec. \ref{sec:U1}. This is mainly because 
the polynomial $Q$ is of increasing degree and thus contains more free parameters for fields of higher dimension. Also for this reason, we will introduce in 
sec. \ref{sec:sect5} another method. 

\subsection{Example: Modular flow of free Fermi field in vacuum (NS)-state}
\label{NSvac}

As an application of these general results, we find the action of the modular flow of a multi-arc $A$ 
for the net $\gA_{\rm Fermi}$ in the vacuum state (Neveu-Schwarz sector), see example 7. 
Even though the free Fermi net is not local but graded local (the free Fermi field $\psi$ has $h=\frac12$) we can easily adapt, in this simple case, our arguments leading to proposition \ref{lem6} to fields obeying Fermi-statistics, i.e. fields of dimension $h \in \frac12 {\mathbb N}$. The main change appears in \eqref{KMS3}, where there is now a pre-factor $-1$ on the right side when $\phi=\psi$
obeys Fermi-statistics. This change propagates to eqs. \eqref{cases1} and \eqref{cases2}, where there now appears a pre-factor $-1$ on the second line on 
the right sides in both equations. Following through this sign change one sees that proposition \ref{lem6} 
 still holds if we replace \eqref{KMS4} by $\widehat{c}_{mn}(t-i) = -\widehat{c}_{nm}(-t)$. 
 
 To determine these functions, we may, in this simple case, test the relation \eqref{constraint} with $p$ points (and with $h=\frac{1}{2}$).
We pick $\zeta, \eta \in \RR$ not equal, and we let $x_l, y_l \in A \subset \bS, k,l=1, \dots, p$ be the pre-images of $\zeta = Z(x_l) \neq \eta= Z(y_k)$, where
$Z=Z_-$ is the function defined by \eqref{Zdef}, and 
where $A$ is the union of $p$ open disjoint arcs in $\bS$ as in \eqref{arcs}. Testing the constraint \eqref{constraint} with these points we get for the 
free Fermi field $\psi$
\ben
\label{constraint1}
-\frac{i}{4\pi}  \frac{Q(x_k,y_l)}{\sinh \pi(\zeta-\eta-i0)} =  \sum_{m,n=0}^{p-1} \widehat{c}_{mn}^{} \left( \zeta - \eta  -i0 \right) (x_k)^m (y_l)^n. 
\een
We note that $v_m(x_k) = (x_k)^m$ and $v_n(y_l) = (y_l)^n$ where $m,n = 0, \dots, p-1$ are $p \times p$ Vandermonde matrices
whose determinants 
\ben\label{vandermonde}
\det [v_j(x_k)] = \prod_{1 \le i < j \le p} (x_i-x_j) \neq 0
\een
do not vanish since all the points 
$x_l, l=1, \dots, p$ are from disjoint intervals in $A$.
Thus, the Vandermonde matrices in 
\eqref{constraint1} may be inverted and therefore $\widehat{c}_{mn}^{}(t)$ is uniquely determined. However, rather than finding the 
coefficients  $Q(x,y) = \sum_{m,n=0}^{p-1} Q_{nm} x^{n} y^m$ from \eqref{Qdef} and inserting the inverses of the Vandermondians directly, 
we may observe that one solution to the constraint \eqref{constraint1}
is of the form
\ben
\widehat{c}_{nm}(t) = -\frac{i}{4\pi} \frac{Q_{nm}}{\sinh \pi t}
\een
and this must hence be the unique solution. It is a good check that this solution is also consistent with the  general properties of proposition \ref{lem6} (for $h=\frac12$).
Substituting the solution into proposition \ref{lem6} (for $h=\frac12$) then gives:

\begin{theorem}
\label{thm:fermi}
For the free massless real Fermi field on $\bS$ and a multi-arc  $A = \cup_{j=1}^p (a_j,b_j) \subset \bS$, the associated modular flow of the Neveu-Schwarz state is 
\ben\label{Gform1}
\langle \Omega_0 | \psi(x) \Delta^{it} \psi(y) | \Omega_0 \rangle = \frac{1}{2\pi i \, (x-y)}
\frac{\Pi_b(x)\Pi_a(y) - \Pi_a(x)\Pi_b(y)}{e^{\pi t} \Pi_b(x)\Pi_a(y) - e^{-\pi t} \Pi_a(x)\Pi_b(y) },
\een
where $x,y \in A$, with the usual boundary value prescription ($y$ approached from the within $\bD^+$, $x$ approached from within $\bD^-$) understood. 
\end{theorem}

Since the action of the modular flow $\sigma^t$ is of second quantized form on the vacuum Hilbert space, it follows that modular flow is uniquely determined by \eqref{Gform1}. We now obtain the generator if the flow, thereby making contact with the 
original analysis due to \cite{Casini0} based on eigenfunctions of the Cauchy kernel. 

First, we transform our result from the circle to the lightray via the Caley transformation $C: \bS \setminus \{+1\} \to \RR, C(x) = -i(x+1)/(x-1)$. 
The lightray fields are then related to the circle fields by $\psi_\bS(x) = \sqrt{C'(x)} \psi_\RR(C(x))$. In terms of the lightray fields, 
eq. \eqref{Gform1} is seen to retain its form, where the arcs of the circle $(a_j,b_j)$ become intervals $(C(a_j),C(b_j))$ of the lightray. 
By abuse of notation, we can thus work with \eqref{Gform1} and pretend that all quantities, such as $x,a_j,b_j,\psi, A$ (see \eqref{arcs}) refer to the lightray. 
Next, we go back from \eqref{Gform1} to the Fourier transform \eqref{Gdef0} using \eqref{fourier}.
This gives us for $x,y \in A \subset \RR$,
\ben
\begin{split}
\langle 0 | \psi(x) (\log \Delta) \psi(y) | 0 \rangle &= \int_\RR s \, \Gamma(s, x, y) \, \dd s \\
&= \frac{Q(x,y)}{4\pi^2 (\Pi_a(x) \Pi_b(x) \Pi_a(y) \Pi_b(y))^{\frac12}} \int_\RR \frac{s e^{-isZ_-(x)+isZ_-(y)}}{1+e^{-s}}  \dd s \\
&= \frac{1}{4\pi^2} \sum_{k=1}^p \int_\RR \frac{s}{1+e^{-s}}  \, \overline{U_{s}^k(x)} U_{s}^k(y)\dd s . 
\end{split}
\een
In the last line we have substituted the functions $U_{s}^k(x) = (-\Pi_a(x) \Pi_b(x))^{-\frac12} q_k(x) e^{isZ_-(x)}$ [compare 
\eqref{eq:Gform}] with the choice $q_k(x) = N_k^{-1} \prod_{i \neq k}(x-a_i)$ for the polynomials, where 
$N_k$ is the constant 
\ben
N_k^2 = -\frac{\prod_{i \neq k} (a_k - a_i)}{\prod_{i=1}^p (a_k - b_i)}
\een
given in \cite{Casini0}, and we have used the identity
$
\sum_{k=1}^p q_{k}(x) q_{k}(y) = Q(x,y)
$
taken from \cite{Casini1} (eq. 2.55). For a function $f \in C^\infty_0(A, \CC)$ of compact support on our multi-interval $A$ \eqref{arcs}, we next 
let $G_A$ be the restriction of the Cauchy kernel $G(x,y)=(1/2\pi i) (x-y-i0)^{-1}$ to $A$, defining an operator on $L^2(A)$
and we let
\ben
k_A(x,y) = (2\pi)^{-2} \sum_{k=1}^p \int_\RR s  \overline{U_{s}^k(x)} U_{s}^k(y)\dd s ,
\een
which will be identified as the modular hamiltonian on the 1-particle space momentarily. 
As shown by \cite{Casini1} (eq. 2.11), we have $G_A U_{k,s} = (1+e^{-s})^{-1} u_{k,s}$, and the functions $U_{s}^k$ in fact 
give a spectral resolution of the operator $G_A$. We can consequently write, with $\psi(f) = \int_{\RR} \psi(x) f(x) \dd x, 
f,g \in C^\infty_0(A, \CC)$:
\ben
\label{eq:prepare}
\langle \Omega_0 | \psi(f) (\log \Delta) \psi(g) | \Omega_0 \rangle = (G_A \bar f, k_Ag) = ( \bar f, G_A k_Ag)
\een
with the usual $L^2$-inner product on the right side.  The kernel $k_A(x,y)$ has been computed in \cite{Casini1} (eqs. 2.72, 2.76), 
with the result ($Z=Z_-$)
\ben
\label{eq:Hdef}
\begin{split}
k_A(x,y) =& -i\bigg(
\frac{1}{Z'(x)}\delta'(x-y) + \frac12 \left[ \frac{1}{Z'(x)} \right]' \delta(x-y) \\
&\qquad -\frac{1}{x-y} \frac{1}{Z'(x)} \sum_{j} \delta[x-y_j(Z(y))] 
\bigg)
\end{split}
\een
where the sum over $k$ is over all pre-images $y_j$ of $Z(y)$ not equal to $y$ itself. Below we will also see that $k_A = \log(G_A^{-1}-1)$. 
\footnote{In the literature \cite{PE}, such formulas have previously been proven for finite-dimensional fermion algebras.}

In order to re-interpret this result on Fock space, it is convenient to give a slightly different, but fully equivalent description 
of the theory $\{ \gA_{\rm Fermi}(I) \}$ on the lightray. 
The $n$-point functions on the lightray are of ``quasi-free'' form in the sense of  \cite{araki_5}
\ben\label{ddtaraki}
\langle \Omega_0 | \psi(f_1) \cdots \psi(f_n) | \Omega_0 \rangle = 
\begin{cases} \sum_{\sigma}
{\rm sgn}(\sigma) \prod_{i=1}^{n/2} (\bar f_{\sigma(i)}, Gf_{\sigma(i+1)}) & \text{$n$ even,} \\
0 & \text{otherwise}
\end{cases}
\een
where the sum is over all perfect matchings in the group of permutations on $n$ elements, and where $G$ is the operator 
defined by the Cauchy kernel. This operator is a projection which in momentum space 
corresponds to the multiplication with the characteristic function on $\RR_+$, i.e. $\widehat{Gf}(k) = 1_{(0,\infty)}(k) \widehat f(k)$. 
As shown in \cite{araki_5}, this leads to an alternative but equivalent description of $\H$ as the fermionic Fock-space  $\cH = \oplus_n \wedge^n \cK$ 
with 1-particle space $\cK=\{ f \in L^2(\RR) \mid \hat f(k) = 0, \forall k \le 0\}$ of square integrable functions $f(x)$ whose Fourier transform $\hat f(k)$ is non-zero only for $k\ge 0$. In terms of this Fock-space, the representation of the light ray fields can be written as 
\ben
\psi(f) = a^*(Gf) + a(G \bar f), 
\een
where $a^*(g), g \in \cK$ are smeared creation operators defined as $a^*(g) |\Psi \rangle = |g \wedge \Psi \rangle$ on any $n$-particle state 
$|\Psi\rangle = |\Psi_1 \wedge \cdots \wedge \Psi_n\rangle \in \wedge^n \cK \subset \H$. 

The 1-particle version of the Reeh-Schlieder theorem implies that it is consistent to introduce on the dense domain $\cD(h_A) = \{ Gf \mid f \in C^\infty_0(A, \CC) \} \subset \cK$  the ``1-particle'' modular Hamiltonian $h_A$ as
\ben
h_A:  \cD(h_A) \to \cK, \quad h_A(Gf) := Gk_Af,  
\een
and as a consequence of \eqref{prepare}, we can then write the modular flow of $\gA_{\rm Fermi}(A)$ for a multi-interval $A=\cup_{j=1}^p I_j$ in second quantized form as 
\ben\label{mflow}
\Delta^{it} = \bigoplus_{n=0}^\infty \bigwedge^n e^{ith_A} . 
\een
In view of \eqref{eq:prepare} 
the final answer may also be (formally) rewritten as
\ben\label{eq:modular_formal}
\langle \Omega_0 | \psi(x) (\log \Delta) \psi(y) \Omega_0 \rangle = \langle \Omega_0 | \psi(x) [H_A, \psi(y)] \Omega_0 \rangle, 
\een
 where $H_A= \half \int_{A \times A} k_A(x,y) \psi(x) \psi(y) \dd x \dd y$ and $k_A$ the kernel of the operator on the right side of \eqref{eq:Hdef}.
Our result for the modular flow 
is thereby seen to be equivalent to the result for the modular flow found previously by \cite{Casini0}. Our arguments therefore in particular provide a rigorous proof of the result by \cite{Casini0}. A rigorous proof in the case that $A$ is a symmetric $p$-interval was previously given by \cite{Rehren:2012wa,Longo:2009mn} using different methods. They also give a (slightly) corrected way to write this equation in exponentiated form, see eq. 4.3 of  \cite{Rehren:2012wa}. 

\subsection{Example: Modular flow of $U(1)$-current on $\bS$}
\label{sec:U1}

The conformal net $\gA_{U(1)}$ for the free $U(1)$ current algebra on the circle
was defined in example 6. Via the Caley transform $C: \bS \setminus \{+1\} \to \RR, C(x) = -i(x+1)/(x-1)$, one obtains 
a corresponding net indexed by open intervals $I \subset \RR$ of the real line (lightray) or a union thereof. 
The circle and lightray currents are related by $j_\bS(x) = C'(x) j_\RR(C(x))$ and thus $j_\RR(x)^* = j_\RR(x)$. The 
corresponding lightray Weyl operators satisfy the same relations as on the circle. The two-point function on the lightray is 
\ben\label{ddtJ}
\langle \Omega_0 | j(x) j(y) | \Omega_0 \rangle = -\frac{1}{2\pi \ (x-y-i0)^{2}}, 
\een
and thus takes the same form as on the circle \eqref{ddt} up to the precise form of the boundary value prescription. 

\medskip

We would next like to understand better the modular flow of the net $\gA_{U(1)}$ of the free $U(1)$ current algebra. 
In so far as proposition \ref{lem6} is concerned, the discussion is actually identical for any bosonic field $\phi$ of dimension $d=1$. 
First we note that for local fields of conformal dimension $h =1,2,3,\dots$ the method used for the free massless Fermi field 
to determine $\widehat c_{mn}$ is inapplicable since the analog of the Vandermonde matrices, $V_l^n = (x_l)^n, l=1,\dots,p, n=1,\dots, 2h(p-1)$,
that now appear in the analogue of \eqref{constraint1} for general $h$
are no longer square matrices and hence not invertible as the sum over $n,m$ would now go up to $2h(p-1)$ according to \eqref{constraint}. 

But we can obtain a weaker result for $h=1$ which will follow instantly from the following two lemmas. 
The first lemma is taken from \cite{Casini1}. 

\begin{lemma}\label{horacio}
Let $x_l \in A \subset \bS, l=1, \dots, p$ ($A$ the union of $p$ open disjoint arcs as in \eqref{arcs}) 
be the pre-images of $\zeta = Z_-(x_l)$ as in \eqref{Zdef}. Then 
\ben\label{want}
\sum_{l=1}^p \frac{1}{2\pi Z'_-(x_l)} \frac{(x_l)^j}{\prod_{n=1}^p (x_l-a_n)(x_l-b_n)} = K_j
\een 
for all natural numbers $j$ in the range $0 \le j \le 2p-2$ and all $\zeta  \in \RR$, where 
\ben\label{Kdef}
K_j = \sum_{l=1}^p \frac{(a_l)^j}{\prod_{n=1}^p (a_l-b_n) \prod_{m \neq l} (a_l-a_m)} = 
-\sum_{l=1}^p \frac{(b_l)^j}{\prod_{n=1}^p (b_l-a_n) \prod_{m \neq l} (b_l-b_m)}
\een
\end{lemma}

From this result, one gets:

\begin{lemma}
Let $x_l, y_l \in A \subset \bS, k,l=1, \dots, p$ be the pre-images of $\zeta = Z_-(x_l) \neq \eta= Z_-(y_k)$ as in \eqref{Zdef}. Then
\ben
\sum_{k,l=1}^p \frac{1}{Z'_-(x_k)} \frac{1}{Z'_-(y_l)} \frac{1}{(2\pi)^2 (x_k-y_l)^2} = 
\frac{p}{[2 \sinh \pi (\zeta- \eta)]^{2}}. 
\een
\end{lemma}

\begin{proof}
Using the notation introduced in the previous proof (with $Z=Z_-$), we have:
\ben\label{11}
\begin{split}
&\sum_{k,l=1}^p \frac{1}{Z'(x_k)} \frac{1}{Z'(y_l)} \frac{1}{(2\pi)^2 (x_k-y_l)^2} \\
=&\sum_{k,l=1}^p \frac{1}{Z'(x_k)} \frac{1}{Z'(y_l)} \frac{(2\pi)^{-2} Q(x_k, y_l)^2}{(\Pi_a(x_k) \Pi_b(y_l)- \Pi_a(y_l)\Pi_b(x_k))^2} \\
=&\frac{(2\pi)^{-2}}{(e^{\pi(\zeta-\eta)}-e^{\pi(\eta-\zeta)})^2} \sum_{k,l=1}^p \frac{1}{Z'(x_k) \Pi_a(x_k)\Pi_b(x_k)} \frac{1}{Z'(y_l) \Pi_a(y_l) \Pi_b(y_l)} Q(x_k, y_l)^2  \\
=&\frac{1}{(2 \sinh \pi(\eta-\zeta))^2} \sum_{k,l=1}^p \frac{Q(a_k, a_l)^2}{
\prod_{n=1}^p (a_l-b_n) \prod_{m \neq l} (a_l-a_m) \prod_{i=1}^p (a_k-b_i) \prod_{j \neq k} (a_l-a_j)
} 
\end{split}
\een
using the previous lemma in the last step. Now it follows from the definition of $Q$ \eqref{Qdef} that 
\ben
Q(a_k, a_l) = \begin{cases}
0 & \text{if $k \neq l$,} \\
-(\Pi'_a \Pi_b - \Pi'_b \Pi_a)(a_k) & \text{if $k = l$,}
\end{cases}
\een
which is also equal to $-\delta_{kl} \prod_{j \neq l}(a_l-a_j) \prod_{i=1}^p (a_l-b_i)$. Inserting this identity into \eqref{11} completes the proof. 
\end{proof}

Now let $\phi$ be a bosonic field of dimension $d=1$. For fixed $\zeta$, consider the pre-images $x_l \in A, l=1, \dots, p$ of $\zeta=Z_-(x_l)$ inside the $p$ 
open disjoint arcs \eqref{arcs}. We can view the $x_l=x_l(\zeta)$ as functions of $\zeta$ and form the operator-valued distribution on $\RR$ given by 
\ben
\label{tilphi}
\widetilde \phi(\zeta) = \frac{1}{2\pi} \sum_{l=1}^p x_l'(\zeta) \phi(x_l(\zeta)), 
\een
formally corresponding to the ``transformation law'' of a primary field of dimension 1.  Our first result on the modular flow is
\begin{theorem}\label{thm:U1thm}
We have for any dimension 1 fields $\phi$ (e.g. the $U(1)$-current) 
\ben
\langle \Omega_0 | \widetilde \phi(\zeta) \Delta^{it} \widetilde \phi(\eta) | \Omega_0 \rangle = \frac{p}{[2 \sinh \pi (t+\zeta- \eta-i0)]^{2}} 
\een
in the sense of distributions in $\eta, \zeta \in \RR$.
\end{theorem}
\begin{proof}
First we apply eq. \eqref{ddt} to $\phi$ and we set $x=x_k, y=y_l$ and sum over $k,l=1, \dots, p$. Then we get from the previous lemma:
\ben
\label{prepare}
\langle \Omega_0 | \widetilde \phi(\zeta) \widetilde \phi(\eta) | \Omega_0 \rangle =  \frac{p}{[2 \sinh \pi (\zeta- \eta-i0)]^{2}} 
\een
Next we set $x=x_k, y=y_l$ in proposition \ref{lem6} and sum over $k,l=1, \dots, p$. Then we get using the notation introduced in the previous proofs
\ben\label{Gform2}
\begin{split}
&\langle \Omega_0 | \widetilde \phi(\zeta) \Delta^{it} \widetilde \phi(\eta) | \Omega_0 \rangle \\
=& (2\pi)^{-2} \sum_{k,l=1}^p \sum_{m,n=0}^{2(p-1)} 
 \frac{\widehat{c}_{mn}^{} \left(t+\zeta- \eta-i0  \right) (x_k)^m (y_l)^n}{Z'(x_k) Z'(y_l) \Pi_a(x_k) \Pi_a(y_l) \Pi_b(x_k) \Pi_b(x_l)} \\
=& \sum_{m,n=0}^{2(p-1)} \widehat{c}_{mn}^{} \left(t+\zeta- \eta-i0  \right) K_m K_n
\end{split}
\een
applying lemma \ref{horacio} in the last step. Now let $f(t) = \sum_{m,n=0}^{2(p-1)} \widehat{c}_{mn}^{}(t) K_m K_n$.  
Comparing \eqref{Gform2} with \eqref{prepare}, we conclude that $f(t-i0) = p [2 \sinh (\pi t-i0)]^{-2}$ for real $t$, and hence for all 
$t$ in the strip $0 > \Im t >-1$ by the edge-of-the-wedge theorem, completing the proof. 
\end{proof}

\begin{remark}
As is well-known, the $U(1)$ current can be represented on the Fock space of two independent free real Fermion fields $\psi_1, \psi_2$
by 
$
j(x) = i:\psi_1\psi_2:(x). 
$
In operator algebraic terms, the $U(1)$-net is a subnet of two copies of the free Fermi net, see \cite{bischoff}. If one could show that 
there was a unit norm vacuum preserving conditional expectation value from the Fermi algebra for region $A$ to the current algebra 
of $A$ (as follows from the work by \cite{bischoff} when $A$ is an interval by Haag duality), then the modular flow on the current algebra $\gA_{U(1)}(A)$
would be that induced by the flow for the Fermi net (thm. \ref{thm:fermi}), by Takesaki's theorem, see e.g. sec. 5 of \cite{accardi}.
At present, however, we do not know that such a conditional expectation exists for multi-component regions $A$. 

As a test, Thm. \ref{thm:fermi}, eq. \eqref{Gform1}, can then be applied to $\widetilde j(\zeta)$ defined as in \eqref{tilphi}. Alternatively, we my compute the left side of \eqref{Gform1} using the modular flow  of the free Fermi field(s) given explicitly in \cite{Casini0, Rehren:2012wa}. One sees after a computation that either results are consistent with
thm. \ref{thm:U1thm}. In the next section, we will again discuss the modular Hamiltonian for the $U(1)$-current from a different perspective.

\end{remark}

\section{Thermal states}
\label{sec:sect5}

It is possible to analyse the modular flow of a thermal state in a similar manner as for the vacuum state. However, we find it useful to use a variation of the 
method described in the previous section which, in essence, corresponds to replacing the matrix elements of the function $\Delta^{it}$ with functions closely related to the 
resolvent, $(\Delta-\lambda)^{-1}$. Such matrix elements will have certain analogous jump properties as the functions $K, H$ introduced above, but their 
exact form depends on the statistics (i.e. conformal dimension $h$) of the field $\phi(x)$. The cases of fermionic\footnote{In the case of fermionic fields, we have graded locality as in the example of the free Fermi field.} and bosonic fields are treated separately in subsecs. \ref{sec:fermions}, \ref{sec:bosons}, respectively. 

We will use the parametrization
$
x = e^{2\pi i u}
$
of the circle. Under this map $A$ consists of intervals 
$\cup_{i=1}^p (a_i, b_i) \subset (0,1)$. $A'$ is as before the interior of the complement.
We 
define, by a slight abuse of notations, 
\ben\label{fieldredef}
\phi(u) \equiv e^{2\pi i h u} \phi(e^{2\pi i u}). 
\een
A Gibbs state is given by the usual formula
\ben
\label{thermal}
\omega_{\beta}(X) = \frac{\tr \left(X e^{-\beta L_0} \right)}{\tr \left(e^{-\beta L_0} \right)}, \quad \Re(\beta)>0, \quad X \in \gA(A),  
\een
where the trace is taken in the vacuum (i.e. the defining) representation $(\pi_0, \H_0)$ of the net. 
We shall mostly take $\beta$ to be real and positive, and occasionally use 
$\tau=i\beta/2\pi$, which is the periodicity of the correlation functions in imaginary direction in the coordinate $u$. 
The general case can be obtained usually by analytic continuation in the end. Then 
it follows immediately that $\omega_{\beta}$ is a $\beta$-KMS-state on $\gA=\gB(\H_0)$ relative to the 1-parameter automorphism group
of rotations of the circle, i.e. translations in $u$. 
By the Reeh-Schlieder theorem, the GNS-vector $|\Omega_\beta \rangle$ corresponding to $\omega_\beta$ is (cyclic and) separating 
for $\gA(A) = \vee_{i=1}^p \gA(I_i)$, and we can define a corresponding modular operator $\Delta \equiv \Delta_{\beta, A}$ as 
in the vacuum situation. 

\subsection{Fermionic fields}
\label{sec:fermions}

We begin by introducing a variant of the construction in the previous section involving resolvents. For $|\Omega\rangle \equiv |\Omega_\beta \rangle$ and 
$\Delta \equiv \Delta_{\beta, A}$, we set\footnote{Here and in the following, we write 
$\phi(u)$ for $\pi_\omega(\phi(u))$, where $\pi_\omega$ is the GNS representation of $\omega \equiv \omega_\beta$.} 
\ben
\label{Fdef}
F(\xi, u, v) = 
\begin{cases}
\langle \Omega | \phi(u) [(\Delta+1)^{-1} + \xi - \half]^{-1} \phi(v) \Omega \rangle & \text{if $-\beta/2\pi < \Im(u)<0$,}\\
-\langle \Omega | \phi(v) [(\Delta^{-1}+1)^{-1} + \xi - \half]^{-1} \phi(u) \Omega \rangle & \text{if $\beta/2\pi > \Im(u)>0.$}
\end{cases}
\een
The resolvents in these expressions are well defined if $\xi \in (\half, \infty) \cup (-\half, -\infty)$ in view of $\Delta >0$, and the analytic continuations 
in $u$ are justified by the KMS property for the Gibbs state, because $\phi(u)|\Omega\rangle$ is a vector-valued holomorphic function 
on the strip $\{ u \in \CC \mid 0<\Im(u)<\beta/2\pi\}$. $\phi$ is assumed to be a hermitian field of conformal dimension $h$ satisfying our assumptions \ref{ass2}. 
The relation to the modular Hamiltonian follows from the  formula 
\ben
\log \Delta= \int_{\half}^\infty \dd \xi \bigg( [(\Delta+1)^{-1} + \xi - \half]^{-1} +  [(\Delta+1)^{-1} - \xi - \half]^{-1} \bigg), 
\een
which trivially holds for positive real numbers $\Delta$ and for positive self adjoint operators in view of the spectral theorem. This immediately gives 
\ben
\label{modularham}
\langle \Omega | \phi(u) (\log \Delta) \phi(v) \Omega \rangle =  \int_{\half}^\infty \dd \xi \bigg( F(\xi, u,v) + F(-\xi,u,v) \bigg), 
\een
where $u$ has a small negative imaginary part.
Thus we should try to find $F$.
Our first lemma is the crucial tool expressing the analyticity/jump properties across the real $u$-axis. 

\begin{lemma}\label{lemma17}
Let $v \in A$ be fixed. If $u \in A'$, then 
\ben
F(\xi, u-i0, v) = F(\xi, u+i0, v). 
\een
If $u \in A$, then 
\ben\label{jump0}
-(\xi+\half) F(\xi, u-i0, v) + (\xi - \half) F(\xi, u+i0, v) = \langle \Omega | \{ \phi(u), \phi(v) \} \Omega \rangle,  
\een
in the sense of distributions. Here $\{ \phi(u), \phi(v) \} = \phi(u) \phi(v)+\phi(v)\phi(u)$ is the anti-commutator. 
\end{lemma}
\begin{proof}
We may formally write $F(\xi, u,v) = \int_\RR [(e^{-s}+1)^{-1} + \xi - \half]^{-1} H(s, u, v) \dd s$, with $H(s,u,v)$ defined by analogy with 
\eqref{cases2}, with ``$-$'' in the second line for fermionic fields. 
Then the statement of the lemma the formally follows from the properties analogous to \eqref{RH}.
This argument is formal because $H(s)$ is defined a priori only after smearing with test functions falling 
off exponentially, see lemma \ref{lemma1}. However, we may also argue directly in the following manner. 
First we let $u,v \in A$ and compute, with $\lambda=\half-\xi$
\ben
\begin{split}
F(\xi,u+i0,v) =& \frac{1}{\lambda} \langle \Omega | \phi(u) (\Delta+1)(\Delta+1-1/\lambda)^{-1} \phi(v) \Omega\rangle\\
=& \frac{1}{\lambda} \langle \Omega | \phi(u) \phi(v) \Omega\rangle + \frac{1}{\lambda^2} \langle \Omega | \phi(u) (\Delta+1-1/\lambda)^{-1} \phi(v) \Omega\rangle.
\end{split}
\een
On the other hand, we also have
\ben
\begin{split}
-F(\xi,u-i0,v) 
=& \frac{1}{\lambda} \langle \Omega | \phi(v) \phi(u) \Omega\rangle + \frac{1}{\lambda^2} \langle \Omega | \phi(v) (\Delta^{-1}+1-1/\lambda)^{-1} \phi(u) \Omega\rangle\\
=& \frac{1}{\lambda} \langle \Omega | \phi(v) \phi(u) \Omega\rangle + \frac{1}{\lambda^2} \langle J \Delta^{\half} \phi(v) \Omega | (\Delta^{-1}+1-1/\lambda)^{-1} J \Delta^{\half} \phi(u) \Omega\rangle\\
=& \frac{1}{\lambda} \langle \Omega | \phi(v) \phi(u) \Omega\rangle + \frac{1}{\lambda^2} \langle \Omega | \phi(u)  \Delta (\Delta+1-1/\lambda)^{-1}  \phi(v) \Omega\rangle.
\end{split}
\een
Here, we used the relations $\phi(u) |\Omega \rangle  = \phi(u)^* |\Omega \rangle = S\phi(u)|\Omega\rangle =  J \Delta^{\half} \phi(u) \Omega\rangle$, the anti-unitary property of $J$, and the property 
$J\Delta^{-1}J=\Delta$, which follow from Tomita-Takesaki theory. Adding up these relations imply 
\ben
(\lambda-1)F(\xi,u+i0,v) - \lambda F(\xi,u-i0,v) = - \langle \Omega | \{ \phi(u), \phi(v) \} \Omega \rangle,
\een
which is equivalent to the statement of the lemma when $u,v \in A$.

In the other case, when $u \in A', v \in A$, we use the formula \eqref{jump11}. On this formula, we act from the left with $\phi(v)$ and from the left with $\phi(u)$ and 
take the expectation value in $|\Omega$. Then we obtain a formula for $F$. Since $u \in A', v \in A$, we have $[\Delta^{-it} \phi(v) \Delta^{it}, \phi(u)]=0$ by Tomita-Takesaki theory, and using this formula to commute the operators inside the expectation values, we get the claim of the lemma when $u \in A', v \in A$.

\end{proof}
 
 As in the previous section, the lemma shows that $F(\xi,u,v)$ defines a function of $u$ for fixed $v \in A$ 
 that is analytic in the cut strip $\{ u \in \CC \mid |\Im(u)|<\beta/2\pi, u \notin A\}$. 
Furthermore, by the KMS condition, it can be checked that $H(s,u-i(\beta-0),v)=-H(s,u+i0,v)$, which is also 
$=-H(s,u+1+i0,v)$ by construction. Thus, $F(\xi,u,v)$ has the same periodicity 
as the 2-point function $\langle \Omega | \phi(u) \phi(v) \Omega \rangle$, that is 
\ben
\label{period}
F(\xi, u+1, v) = F(\xi,u+\tau, v) = -F(\xi,u,v).
\een
This allows us to define $F$ as a function of $u$ on the entire complex plane cut by $A + \ZZ + \tau \ZZ$.
The limits from below the real axis define hermitian distributional kernels $F(\xi, u- i0, v)$ that are of positive/negative 
type for $\xi>\half$ resp. $\xi<-\half$, which follows from $\Delta>0$; similarly for the limit from above the real axis. 
Similar statements hold when fixing $u \in A$ and viewing $F(\xi,u,v)$ as a function of $v$.

As in the vacuum case studied in sec. \ref{localfields}, we would next like to have a result like \eqref{firstbound}
of lemma \ref{lemma1} about the potential singularities of $F$ at the end points of the intervals. 
Unfortunately, the proof strategy of lemma \ref{lemma1} does not hold in 
the present case since we have no analogue of the Hislop-Longo theorem for thermal states. 

Instead, we will prove first a result comparing the modular operator of a thermal state $\omega_\beta$ 
for the full algebra $\gA = \gB(\H_0)$ to the modular operator 
for the partial algebra $\gA(A)$. The point is that the former corresponds to rotations of the circle (i.e. translations of the coordinate $u$)
and is thus known. The idea is more precisely to apply eq. \eqref{frede} 
to the case $\gM_1= \gA$, $\gM_2 = \gA(A)$ (viewed as operator algebras on the GNS-Hilbert space $\H_\beta \equiv \H$ of the thermal 
state $\omega_\beta \equiv \omega$), so 
we can define the modular operators $\Delta_i$ of $\gM_i$ on $\H$ with respect to the cyclic and separating vector $|\Omega\rangle \equiv |\Omega_\omega\rangle$.

\begin{lemma}
\label{lemma14}
Let $s \in \RR$, let $u,v \in A = \cup_{i=1}^p (a_i,b_i)$. Then
\ben
\begin{split}
&\Big|
\langle \Omega |  \phi(u) (1+e^s \Delta_{1})^{-1} \phi(v) \Omega \rangle -
\langle \Omega |  \phi(u) (1+e^s \Delta_{2})^{-1} \phi(v) \Omega \rangle
\Big| \\
& \qquad \lesssim  \sum_{q_j}  |u-q_j|^{-h} |v-q_j|^{-h} , 
\end{split}
\een
with implicit constant depending on $s$ and the endpoints $\{q_j\}$ of the intervals.
\end{lemma}

\begin{proof} 
Let $X,Y \in \gM_2$. From \eqref{frede}, we get with the notations introduced around \eqref{H1}
\ben\label{frede1}
\begin{split}
&\langle \Omega |Y^* (1 + e^s\Delta_1)^{-1} X  | \Omega \rangle -  \langle \Omega | Y^* (1 + e^s\Delta_2)^{-1} X  | \Omega \rangle\\
=&  \langle (1- P_2) I^{-1}  (1 + e^s\Delta_1)^{-1} Y \Omega | (1- P_2) I^{-1} (1 + e^s\Delta_1)^{-1} X \Omega \rangle \ .
\end{split}
\een
For $y>0$, we can write
\ben\label{fourier}
\frac{1}{1+y} 
=\frac{i}{2} \int_\RR \frac{y^{it}}{\sinh[\pi (t+i0)]} \dd t  .
\een
Therefore, by the spectral calculus applied to $y=e^s \Delta_1$,
\ben
(1- P_2) I^{-1} (1 + e^s\Delta_1)^{-1} X |\Omega \rangle = \frac{i}{2}  \, (1- P_2) I^{-1} 
\int_\RR \dd t \ \frac{e^{ist} \Delta_1^{it}}{\sinh[\pi (t+i0)]}  X |\Omega \rangle . 
\een
The key idea is now the following.
Suppose that, for $|t| < t_0$ and some $t_0>0$, we knew that $\sigma_1^t(X)$ is in $\gM_2$, so
$\Delta_1^{it} X |\Omega \rangle=\sigma_1^t(X) |\Omega \rangle$ is in the domain 
${\mathcal D}(S_2)$, so $I^{-1} \Delta_1^{it} X |\Omega \rangle$ is in $\H_2$, so $(1- P_2) I^{-1}  \Delta_1^{it} X |\Omega \rangle=0$. 
Then we can effectively restrict the range in the integral to $|t| \ge t_0$ and drop the $i0$-prescription, and this is the moral reason for the existence of our bound. 

An even better estimate is obtained if instead we choose a, say even, 
real-valued smooth function $h(t) \ge 0$ such that $h(t) = 0$ for $|t|<\frac12$ and $h(t) = 1$ for $|t| \ge 1$, say, and write 
\ben
\label{ekdel}
(1- P_2) I^{-1} (1 + e^s\Delta_1)^{-1} X |\Omega \rangle = \frac{i}{2} 
\, (1- P_2) I^{-1} \int_\RR  \dd t \ \frac{e^{ist} \Delta_1^{it}}{\sinh(\pi t)} h \left( \frac{t}{t_0} \right)  X |\Omega \rangle . 
\een
Now we take a test function $f$ compactly supported inside $A$ and we denote the distance of the support of $f$ to the boundary of $A$ by
$\delta={\rm dist}(\partial A, {\rm supp}(f))$; of course $\delta>0$. 
Furthermore, we let $\{X_n\}$ be a sequence in $\gM_2$
converging strongly to $\phi(f)$. Such a sequence exists since we assume that the local fields are affiliated. Then \eqref{ekdel} also holds for $X=\phi(f)$.
We next wish to use the mode expansion \eqref{smeared} inside \eqref{ekdel}. Because $\omega \equiv \omega_\beta$ is a $\beta$-KMS state on $\gM_1$ with respect to the automorphic actions of rotations $\alpha_t=\alpha_{g_t}, g_t(z) = e^{-i\beta t} z$, the modular flow $\Delta_1^{it}$ corresponds to rotations in the sense that $\Delta^{it}_1 \phi(z) \Delta_1^{-it} = e^{-iht\beta} \phi(e^{-i\beta t}z)$ for $z \in \bS$. Then it is clear that $\Delta^{it}_1 \phi(f) \Delta_1^{-it}$ will remain affiliated with $\gM_2$ as long as $|t|<\delta/\beta$, 
so \eqref{ekdel} holds with $X=\phi(f)$ and $t_0=\delta/\beta$.

In terms of modes
\ben
\Delta_1^{it} \phi_n \Delta_1^{-it} = e^{i\beta (n+h) t} \phi_n .
\een
In particular, since the modular flow preserves the state, we must have $\langle \phi_m \Omega | \phi_n \Omega \rangle
= \delta_{n,m} \| \phi_n \Omega \|^2 = \delta_{n,m} \omega( \phi_n^* \phi_n^{})$. From (iv) of assumption \ref{ass2} it follows that $\phi_n^* \phi_n^{} \lesssim (1+|n|)^{2h-1} (1+L_0)^{2k}$ for some 
$k \ge 0$, and therefore we must have 
\ben 
\| \phi_n \Omega \|^2   \lesssim (1+|n|)^{2h-1} \omega[(1+L_0)^{2k}]
\een
for all $n$. Since $\omega\equiv \omega_\beta$ is a Gibbs state, we conclude
\ben
\| \phi_n \Omega \|^2 
\lesssim (1+|n|)^{2h-1} \frac{\tr [(1+L_0)^{2k} e^{-\beta L_0} ]}{\tr (e^{-\beta L_0} )}
\lesssim (1+|n|)^{2h-1}.
\een
Using the definition of $I$ as well as $1-P_2 \le 1$, we get for any complex constants $c_n$, 
\ben
\begin{split}
\left\| \sum_{n \in \bZ} c_n (1- P_2) I^{-1} \phi_n \Omega \right\|^2 
&\le  \sum_{n,m \in \bZ}  \bar c_n c_m  \, \langle I^{-1} \phi_m \Omega | I^{-1} \phi_n \Omega \rangle  \\
&=  \sum_{n,m \in \bZ}  \bar c_n c_m  \, \langle \phi_m \Omega | (1+e^s \Delta_1) \phi_n \Omega \rangle  \\
&=  \sum_{n \in \bZ}  |c_n|^2  \, ( \| \phi_n \Omega \|^2 +  e^s \| \phi_n^* \Omega \|^2 ) \\
&=  \sum_{n \in \bZ}  |c_n|^2  \, ( \| \phi_n \Omega \|^2 +  e^s \| \phi_{-n} \Omega \|^2 ) \\
&\le (1+e^s) \sum_{n}  |c_n|^2 \, O(|n|^{2h-1}) . 
\end{split}
\een
Taking the norm squared of \eqref{ekdel} then gives:
\ben
\label{ekdel1}
\begin{split}
& \| (1- P_2) I^{-1} (1 + e^s\Delta_1)^{-1} \phi(f) \Omega \|^2 \\
=& \frac{1}{4} \left\| \sum_{n \in \bZ}
\left( \int_\RR  \frac{e^{ist+i \beta (n+h) t}}{\sinh(\pi t)} h \left( \frac{\beta t}{\delta} \right)  \dd t \right)
 f_n  (1- P_2) I^{-1} \phi_n \Omega \right\|^2 \\
\le & (1+e^s)  \sum_{n \in \bZ}
\left| \int_\RR  \frac{e^{ist+i\beta (n+h) t}}{\sinh(\pi t)} h \left( \frac{\beta t}{\delta} \right)  \dd t \right|^2 \, 
 |f_n|^2 \, O(|n|^{2h-1})  \\
\equiv & (1+e^s)
\sum_{n \in \bZ} \varphi_{\delta/\beta} \bigg( s +\beta (n+h)  \bigg) \, |f_n|^2 \, O(|n|^{2h-1}) \,   ,
\end{split}
\een
where $f_n=\int_0^1 e^{2\pi i n u} f(u) \dd u$ the Fourier components of $f(u)$, and where $\varphi_{t_0}(s)$ is some smooth function 
which can be chosen to satisfy for large $|s|$ a bound of the form $|\varphi_{t_0}(s)| \le O[ \, (1+t_0 |s| )^{-N} ]$ for as large an $N$ as we wish.
The Fourier coefficients are trivially bounded by the $L^1$-norm of $f$. 
The bound \eqref{ekdel1} then gives, altogether
\ben
\label{ekdel2}
\begin{split}
& \| (1- P_2) I^{-1} (1 + e^s\Delta_1)^{-1} \phi(f) \Omega \| \\
\lesssim &  \| f \|_{L^1} \left(
\sum_{n \in \bZ} O( |\delta n|^{-N}) \, O(|n|^{2h-1})  \right)^{\frac12} \\
\lesssim & \| f \|_{L^1} \left( \frac{ 1}{\delta} \right)^{h} ,
\end{split}
\een
with implicit constants depending on $s$.
We can likewise finde a sequence $\{Y_n\}$ be a sequence in $\gM_2$
converging strongly to $\phi(g)$, and thereby obtain a similar result as \eqref{ekdel2} replacing $\phi(f)$ by $\phi(g)$.
Combining these two results now with \eqref{frede} and using the Cauchy-Schwarz inequality on the right side of that equation gives
\ben
\begin{split}
&\Big|
\langle \Omega | \phi(g)^* (1 + e^s\Delta_1)^{-1} \phi(f)  | \Omega \rangle -  \langle \Omega | \phi(g)^* (1 + e^s\Delta_2)^{-1} \phi(f)  | \Omega \rangle
\Big|
\\
\lesssim &  \| f \|_{L^1} \| g \|_{L^1} \left( \delta_f \delta_g \right)^{-h}. 
\end{split}
\een
Letting $f,g$ tend to delta-distributions centered at $u,v$ then by defintion, $\delta_f \to {\rm dist}(u, \partial A) \lesssim \sum |u-a_j| |u-b_j|$ 
and likewise for $g$, and the $L^1$ norms remain bounded. This gives the claim of the lemma.
\end{proof}

From here, we can get:

\begin{lemma}
\label{lemma15}
Let $u,v \in A = \cup_{i=1}^p (a_i,b_i)$. Then
\ben
\Big|
F(\xi, u \mp i0, v)
\Big|  \lesssim    \sum_{q_j}  |u-q_j|^{-h} |v-q_j|^{-h}  + |u-v|^{-2h}, 
\een
with implicit constant depending on $\beta,s$ and the endpoints $\{q_j\}$ of the intervals.
\end{lemma}

\begin{proof}
For definiteness, consider $F(\xi, u - i0, v)$.
From \eqref{fourier}, we get for $\Delta>0$ the identity
\ben
\label{jump11}
\frac{1}{(\Delta+1)^{-1} + \xi-\half} = \frac{1}{\xi+\half} - \frac{i}{2} \frac{1}{\xi^2-\tfrac{1}{4}} \int_{-\infty}^\infty \dd t \left(
\frac{\xi-\half}{\xi+\half}
\right)^{it} \frac{\Delta^{it}}{\sinh[\pi(t-i0)]}.
\een
We use this for $\Delta=\Delta_1$, the modular operator for the full algebra $\gA$. Then $\Delta_1^{it}$ generates translations by $-\beta t/2\pi$ of the coordinate $u$. 
Now we sandwich the above identity between $\langle \Omega | \phi(u-i0)$ and $\phi(v+i0)|\Omega\rangle$. Then if $F_1$ is defined as $F$ but with $\Delta_1$,
we get 
 \ben
 \begin{split}
F_1(\xi,u-i0,v) =& \frac{1}{\xi+\half} \langle \Omega | \phi(u)\phi(v) \Omega\rangle \\
&
- \frac{i}{2} \frac{1}{\xi^2-\tfrac{1}{4}} \int_{-\infty}^\infty \dd t \left(
\frac{\xi-\half}{\xi+\half}
\right)^{it} \frac{\langle \Omega | \phi(u+\beta(t-i0)/2\pi) \phi(v) \Omega\rangle}{\sinh[\pi(t-i0)]}.
\end{split}
\een
The second term is uniformly bounded in $u,v$ since the integrand is the boundary value of an analytic function 
that his holomorphic for $-\epsilon<\Im(t)<0$ with algebraic singularities. The first term behaves as $\sim (u-v)^{-2h}$, since this is the UV-behavior 
of a thermal 2-point function of a field of conformal dimension $h$, as one may also prove rigorously by decomposing the field 
into modes and using (iv) of assumption \ref{ass2}. The proof now follows by writing $|F| \equiv |F_2| \le |F_2-F_1|+ |F_1|$ and 
using the preceeding lemma \ref{lemma14} on $|F_2-F_1|$. We remark that the proof actually shows $F(\xi,u-i0,v) \sim \frac{e^{-i\pi h}}{2\pi } \frac{1}{\xi+\half} (u-v-i0)^{-2h}$
for $u$ near $v$ if the field $\phi$ is in standard normalization \eqref{ddt}.

\end{proof}

Our assumptions \ref{ass2} and the fermionic nature of $\phi$ (i.e. $h \in \half {\mathbb N}$) imply the operator expansion in anti-commutator form,
\ben
\{ \phi(u), \phi(v) \} = \sum_{n=0}^{2h-1} \delta^{(n)}(u-v) O_n(\half(u+v)), 
\een
as one can see for instance using the relation between our operator algebraic formalism and Vertex Operator Algebras, see \cite{longo2}.
Here the $O_n$ are hermitian bosonic primary fields of conformal dimension $2h-1-n$, and $\delta^{(n)}(u) 
= \frac{\dd^n}{\dd u^n} \delta(u)$. We let $\langle O_n \rangle = \langle \Omega | O_n(u) \Omega\rangle$, which is independent of 
$u$ for our rotation invariant thermal state $|\Omega \rangle$, and use this relation in lemma \ref{lemma17}, to obtain a new version of 
the jump condition involving the thermal expectation values $\langle O_n \rangle$. 

We summarize the properties of $F$ \eqref{Fdef}:
\begin{theorem}\label{thm3}
Let $\phi$ be a fermionic field of conformal dimension $h$ in the normalization \eqref{ddt}.
\begin{itemize}

\item
For fixed $v \in A$ and $\xi \in (\half, \infty) \cup (-\half, -\infty)$, 
$F(\xi, u,v)$ is a holomorphic function of $u$ on the periodically cut complex plane $\CC \setminus (A + \ZZ + \tau \ZZ)$
with the periodicity $F(\xi, u+1, v) = F(\xi,u+\tau, v) = -F(\xi,u,v)$.

\item
Across the cuts, $u \in A$, $F$ satisfies the jump condition 
\ben
\label{eq:jump}
-(\xi+\half) F(\xi, u-i0, v) + (\xi - \half) F(\xi, u+i0, v)=\sum_{n=0}^{2h-1} \delta^{(n)}(u-v) \langle O_n \rangle.
\een

\item
Near the end-points $q_i$ of the intervals, $|F| \lesssim |q_i-u|^{-h}$, as a function of $u$.

\item
Near $v$, 
\ben
\label{pole}
F(\xi, u- i0, v) \sim \frac{e^{-i\pi h}}{2\pi } \frac{1}{\xi+ \half} (u-v-i0)^{-2h} 
\een
as a function of $u$.

\item
Analogous properties hold true for $u \leftrightarrow v$. 
\end{itemize}
\end{theorem}

Now we let 
\ben
\label{fermitwopoint1}
G(u-v)=\frac{\vartheta_1'(\tau)}{2\pi i \vartheta_{3}(\tau)} \frac{\vartheta_{3}(u-v;\tau)}{\vartheta_1(u-v;\tau)}
\een
also equal to the thermal 2-point function of the free Fermi field (see below). Our conventions for the Jacobi $\vartheta$ functions are 
summarized in the appendix.  
Then the combination $G(u-u') F(\xi,u',v)$ is doubly periodic in $u'$ for any fixed $v\in A$. Thus, integrating around a contour $\gamma_\square$ surrounding the fundamental 
parallelogram as in fig. \ref{fig:1}, we get zero, 
\ben
\label{eq:netzero}
0=\oint_{\gamma_\square} \dd u' \, G(u-u') F(\xi,u',v) .  
\een

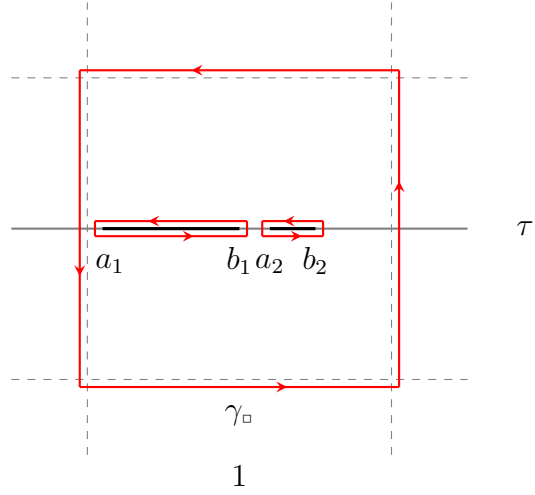
\begin{figure}
\centering
\begin{tikzpicture}[scale=1.0]
\draw[gray, dashed] (-1,0) -- (5,0);
\draw[gray, dashed] (-1,4) -- (5,4);
\draw[gray, dashed] (0,-1) -- (0,5);
\draw[gray, dashed] (4,-1) -- (4,5);
\draw (2,-1) node[below]{$1$}; 
\draw (5.5,2) node[right]{$\tau$}; 
\draw[gray, thick] (-1,2) -- (5,2);
\draw[black, very thick] (.2,2) -- (2,2);
\draw[black, very thick] (2.4,2) -- (3,2);
\draw (.3,1.8) node[below]{$a_1$}; 
\draw (2,1.9) node[below]{$b_1$}; 
\draw (2.4,1.8) node[below]{$a_2$}; 
\draw (3,1.9) node[below]{$b_2$}; 
\draw[red, directed, thick] (-.1,-.1) -- (4.1,-.1);
\draw[red, directed, thick] (4.1,-.1) -- (4.1,4.1);
\draw[red, directed, thick] (4.1,4.1) -- (-.1,4.1);
\draw[red, directed, thick] (-.1,4.1) -- (-.1,-.1);
\draw (2,-.2) node[below]{$\gamma_\square$};
\draw[red, directed, thick] (.1,1.9) -- (2.1,1.9);
\draw[red, thick] (2.1,1.9) -- (2.1,2.1);
\draw[red, directed, thick] (2.1,2.1) -- (.1,2.1);
\draw[red, thick] (.1,2.1) -- (.1,1.9);

\draw[red, directed, thick] (2.3,1.9) -- (3.1,1.9);
\draw[red, thick] (3.1,1.9) -- (3.1,2.1);
\draw[red, directed, thick] (3.1,2.1) -- (2.3,2.1);
\draw[red, thick] (2.3,2.1) -- (2.3,1.9);
\end{tikzpicture}
\caption{Illustration of the contour deformation in complex $u$-plane.}
\label{fig:1}
\end{figure}

Since $G(u-v) \sim \frac{1}{2\pi i} (u-v)^{-1}$ for $u \sim v$,
and $u=v$ is the only singularity of $G$ inside the fundamental parallelogram, we can deform the contour $\gamma_\square$ to a contour tightly surrounding the intervals 
as shown in fig. \ref{fig:1}, for $\Im(u) \to 0^-$ and $\Re(u) \in A$. 
Then we collect the residue, use the jump condition in the theorem, as well as $\sum_{n=0}^{2h-1}  \langle O_n \rangle  \, G^{(n)}(u-v-i0) = \langle \Omega
 | \phi(u) \phi(v) \Omega\rangle$.This gives:

\begin{corollary}\label{cor:2}
Let $\phi$ be a fermionic primary field of dimension $h$.
$F$ defined as in \eqref{Fdef} satisfies (in the distributional sense for $u,v \in A$):
\ben
\begin{split}
\langle \Omega_\beta | \phi(u) \phi(v) \Omega_\beta \rangle =& (\xi-\half)  F(\xi,u-i0,v) \\
& + \frac{1}{2\pi i} \bint_A \dd u' \, \frac{\vartheta_1'(\tau)}{\vartheta_{3}(\tau)} \frac{\vartheta_{3}(u-u'-i0;\tau)}{\vartheta_1(u-u'-i0;\tau)} F(\xi,u'-i0,v) .
\end{split}
\een
Here, 
$\bint$ denotes a regulated version of the integral described in the following remark.
\end{corollary}

\noindent
{\bf Remark/Definition:}
The meaning of the regularized integral operation $\bint_A$ is somewhat subtle because because  
$F^\pm(\xi, u, v) = F(\xi, u \pm i0, v)$ are distributions of $u \in \RR$ for fixed $v \in A$. 
As such, there is no obvious meaning to the integral in cor. \ref{cor:2}. To define $\bint_A$, we 
must remember the origin of the formula from an integration along a contour fitting tightly around $A$, see 
fig. \ref{fig:1}. We consider tempered distributions $\widetilde F^\pm(\xi, u, v)$ of $u \in \RR$ with the property that 
\begin{enumerate}
\item[(a)] $\widetilde F^+(\xi, u, v) - \widetilde F^-(\xi, u, v) = F^+(\xi, u, v) - F^-(\xi, u, v)$, 

\item[(b)] $\widetilde F^\pm(\xi, u, v) |_{A'}=0$, as a distribution tested with functions of $u$ having compact support in 
$A'$ (thus away from any boundary point), and 

\item[(c)] $\widetilde F^\pm(\xi, u, v) |_{A}=0$ satisfies the jump condition \eqref{eq:jump} of thm. \ref{thm3}
as a distribution tested with functions of $u$ having compact support in 
$A$ (thus away from any boundary point). 
\end{enumerate}
We then define the regulated integral by replacing $F^-$ with $\widetilde F^-$, now integrated over all $u$, which is 
now well-defined. Cor. \ref{cor:2} holds with this prescription, since (a), (b), (c) are all that is used about $F$ in the vicinity of 
the real axis in the above proof. 

The existence of $\widetilde F^\pm$ can be seen as follows. In $A$, it is defined to be equal
to $F^\pm$. In $A'$, it is defined to zero. This defines distributions for all $u \in \RR$, except for the boundary points $q_i$. 
We wish to define $\widetilde F^\pm$ by a suitable distributional extension. At the boundary points the unextended 
distributions have ``scaling degree'' at most $-h$ (by thm. \ref{thm3}), i.e. finite. Therefore, by standard results (see e.g. \cite{hollands0}), 
an extension exists and is unique up to addition of $\sum_{q_j} \sum_{k \le h-1} {}^\pm A_j^k \delta^{(k)}(u-q_j)$. 
The free parameters must be adjusted in such a way that (c) holds for all $u \in \RR$, not just inside $A$ away from 
the boundary points, and such that $(a)$ still holds. This uniquely determines $\widetilde F^\pm$, and hence our prescription 
$\bint_A$.\footnote{ 
We remark that in the expressions considered below in sec. \ref{sec:U1current}, 
our regulated integral has the same meaning as that explained in sec. 4 of \cite{Casini1}.}

\medskip
\noindent
 The corollary imposes on $F$ an integral equation at 
which we will look at in more detail in the following example. The example also suggests that corollary \ref{cor:2} 
can be used to find the modular operator in the case of a general fermionic primary field of dimension $h$. 

\subsection{Example: Modular Hamiltonian for thermal state of free Fermi field}
\label{sec:ModFermbeta}

Now we analyze what we can learn from thm. \ref{thm3} and cor. \ref{cor:2} for the free Fermi field, where $h=\half$ and $\{\psi(u), \psi(v) \} = \delta(u-v)1$. 
For this theory, we can actually prescribe thermal states \eqref{thermal} using either the Neveu-Schwarz or the Ramond 
representations. For the Neveu-Schwarz (vacuum) sector, the representation is the half-integer moded 
expansion \eqref{NS}, now denoted as $\psi_{\rm NS}(z)$. 
For the Ramond-sector, the representation corresponds to an integer moded expansion  given in \eqref{R} now denoted 
as $\psi_{\rm R}(z)$. The Gibbs 
states \eqref{thermal} in the Neveu-Schwarz/Ramond representation correspond, accordingly, to the thermal state vectors $|\Omega_{{\rm NS},\beta}\rangle$ and
$|\Omega_{{\rm R},\beta}\rangle$, respectively. The thermal 2-point functions are \cite{zuber}:
\ben
\label{fermitwopoint}
\langle \Omega_{{\rm X},\beta} | \psi_{\rm X}(u) \psi_{\rm X}(v) \Omega_{{\rm X},\beta} \rangle 
= G_{\rm X} (u-v-i0)
\een
 where the 
subscript ${\rm X} \in \{ {\rm R, NS} \}$ indicates the choice of boundary conditions (sector).
For NS, the definition was already given above in \eqref{fermitwopoint1}, whereas for R, we 
set 
\ben
\label{fermitwopoint1}
G_{\rm R}(u-v)=\frac{\vartheta_1'(\tau)}{2\pi i \vartheta_{2}(\tau)} \frac{\vartheta_{2}(u-v;\tau)}{\vartheta_1(u-v;\tau)}
\een
Next we define $F_{\rm X}$ as in \eqref{Fdef}. By thm. \ref{thm3} we know its properties, with the only trivial difference that $F_{\rm R}(s, u+1, v)=F_{\rm R}(\xi, u, v)$ in the Ramond sector, where the 
subscript ${\rm X} \in \{ {\rm R, NS} \}$ indicates the choice of boundary conditions.
These properties uniquely determine $F_{\rm X}$ in our case (cf. table \ref{tab:1}):
\ben
\begin{split}
F_{\rm R,NS}(\xi, u, v) =& \frac{1}{\xi+\half} \frac{\vartheta_1'(\tau)}{2\pi i \vartheta_1(u-v; \tau)} \frac{\vartheta_{2,3}(u-v- is|A|/2\pi;\tau)}{\vartheta_{2,3}(- is|A|/2\pi;\tau)} \\
& \times
\left( \frac{\Pi_a(u;\tau) \Pi_b(v;\tau)}{\Pi_b(u;\tau) \Pi_a(v;\tau)} \right)^{is/2\pi},
\end{split}
\een
assuming for instance $-\epsilon<\Im(u)<0$. Here $|A|=\sum_j (b_j-a_j)$ is the total length of the intervals,
\ben
\label{sdef}
e^s = \frac{\xi+\half}{\xi-\half}, 
\een
and
\ben
\Pi_a(u;\tau) = \prod_{i=1}^p\vartheta_1(u-a_i;\tau)
\een 
etc. 

By the well-known periodicity properties and pole/zero structure of the $\vartheta$-functions, it is easy to check that the above formula indeed satisfies the 
the properties known to hold by thm. \ref{thm3}. For instance, the role of the second line is to give the required jump across $A$, and the role of the pole at $u=v$ in the first line 
gives the delta-function in \eqref{jump11}, also consistent with \eqref{pole}. To show that $F$ as given is the only solution to the properties in 
thm.  \ref{thm3}, we multiply any solution $F$ with the inverse of the second line. This will cancel the jump of $F$ except for the delta function in \eqref{jump11}. 
Thus, the multiplied $F$ can have a first order pole at $u=v$ mod $\ZZ$ with residue $1/(\xi- \half)$, by \eqref{jump11}.  It cannot have poles, however, 
at the end-points $q_i$ of the intervals, since thm.  \ref{thm3} tells us that $|F| \lesssim |q_i-u|^{-1/2}$. Thus, the multiplied function $F$ is meromorphic with a 
single pole in the fundamental parallelogram, with the same residue and periodicity properties as the first line. There can only be one such function, by 
well-known results on elliptic functions. 

Knowing $F$, we get the kernel of the modular hamiltonian $\langle \Omega_{{\rm X},\beta} | \psi_{\rm X}(u) (\log \Delta) \psi_{\rm X}(v) \Omega_{{\rm X},\beta} \rangle$ by integrating up 
\eqref{modularham}. Thus, we have, in principle, found this kernel. The (non-trivial) integration over $\xi$ in \eqref{modularham} has recently performed 
by \cite{blanco,fries}, who have independently arrived at an analogous formula as for $F$ for a closely related quantity, 
by a different method based on resolvents. It is instructive to see more precisely 
how this method is related to ours, as it may also shed light on how to proceed in the case of fields with general conformal dimension $h$. 

For this purpose, let us set $G_{{\rm X}}(u-v)$ be the thermal 2-point function given by \eqref{fermitwopoint}, with ${{\rm X}} \in \{ {\rm NS,R}\}$. 
By thm. \ref{thm3}, $F_{{\rm X}}$ has the 
same periodicity as $G_{{\rm X}}$, so the combination $G_{{\rm X}}(u-v') F_{{\rm X}}(\xi,u,v)$ is doubly periodic in $u$ for any fixed $v\in A$. Thus, integrating again around a contour $\gamma_\square$ surrounding the fundamental 
parallelogram as in fig. \ref{fig:1}, we get zero as in \eqref{eq:netzero}, with $G=G_X$ now in that equation.  
Such an identity also derived by \cite{fries} via a different argument using the special properties of the free Fermi field. 
By contrast, we have so far only used general properties and thm. \ref{thm3}, so our argument works for any CFT. 

Now, we deform again the contour $\gamma_\square$ to a contour tightly surrounding the cuts $A$ inside the fundamental 
parallelogram, see fig. \ref{fig:1}, use the jump properties and bounds on $F$ as given in thm. \ref{thm3} and collect the residue. 
While these steps can still be performed in any CFT according to thm. \ref{thm3}, the free Fermi case is especially simple because the 
poles of $F$ near the end-points $q_i$ of the intervals are of the order $|u-q_i|^{\half}$, and hence integrable.
Thereby, we immediately get the integral equation
\ben
G_{{\rm X}}(v'-v-i0) = (\xi-\half)F_{{\rm X}}(\xi,u-i0,v) + \int_A \dd u' \, G_{{\rm X}}(u-u'-i0)  F_{{\rm X}}(\xi,u'-i0,v) 
\een
without any need for a regulator in the integral as in cor. \ref{cor:2}.
Dropping the subscript ``$_{{\rm X}}$'' for ease of notation and indicating by $G_A \equiv G|_A$ the restriction of the kernels \eqref{fermitwopoint} to $A \times A$, 
we may view $G_A$ as an operator $L^2(A) \to L^2(A)$. In fact, the anti-commutation relations and the positive nature of the 
Hilbert space inner product in \eqref{fermitwopoint} imply the operator inequalities $0 \le G_A \le 1$. Now we define $F^\pm_A(\xi)$ the operator
defined by the restriction $F(\xi,u\pm i0,v)$ to $A$ in kernel notation, i.e. when 
$u,v \in A$. 
We can then also write the above integral equation in operator notation as $G_A =(G_A + \xi - \half)F^-_A(\xi)$ 
and solve it as 
\ben
F^-_A(\xi)=(G_A + \xi - \half)^{-1}G_A
\een
when $\xi \in (\half, \infty) \cup (-\half, -\infty)$. 
In view of \eqref{modularham}, we should then further calculate, in operator notation  
\ben
\int_{\half}^\infty \dd \xi \bigg( F^-_A(\xi) + F^-_A(-\xi) \bigg) 
=  \log(G_A^{-1}-1)G_A. 
\een
Thus, as a kernel, the modular hamiltonian is given by 
\ben
\label{peschl1}
\langle \Omega | \psi(u) (\log \Delta) \psi(v) \Omega \rangle = \bigg(  \log(G_A^{-1}-1)G_A \bigg)(u,v). 
\een
Formally, this may also be expressed as saying $H_A= \half \int_{A \times A} k_A(x,y) \psi(x) \psi(y) \dd x \dd y$ and $k_A=\log(G_A^{-1}-1)$ is 
the hamiltonian appearing in eq. \eqref{eq:modular_formal}.
This is in accord with a well known formula \cite{PE1} for the modular Hamiltonian of the free Fermi field in terms the restricted 2-point function $G_A$, which can be derived 
using the Fock-space structure of theory. In the present case, however, our derivation was rather different and can be paralleled for fermionic fields in general CFTs which are not necessarily equivalent to free field theories. As derived here, our formula holds for thermal states in either 
the NS or the R sector, choosing for $G$ either one of the 2-point functions \eqref{fermitwopoint}. Thus, apart from confirming our method, this relation to integral equations might also be useful in the case of more general CFTs where one does not have a priori relations like \eqref{peschl1}.

\subsection{Bosonic fields}
\label{sec:bosons}

We now repeat a similar analysis for a bosonic field $\phi$ of dimension $h \in {\mathbb N}$, satisfying our standing assumptions \ref{ass2}. 
In fact, most results in sec. \ref{sec:fermions} for fermions hold just as well for bosons, with 
nearly identical proof. A difference, however, arises in the jump condition \eqref{jump0}, which involves the anti-commutator of $\phi$, whereas for bosons, we would like to have the commutator. This means that we should work with a different definition of $F$, \eqref{Fdef}. In the bosonic case, we instead define 
using the shorthands $|\Omega\rangle \equiv |\Omega_\beta \rangle$ and 
$\Delta \equiv \Delta_{\beta, A}$
\ben
\label{Fdef1}
F(s, u, v) = 
\begin{cases}
\langle \Omega | \phi(u) [1-e^s (1-\Delta)^{-1}]^{-1} \phi(v) \Omega \rangle & \text{if $-\beta < \Im(u)<0$,}\\
\langle \Omega | \phi(v) [1-e^s(1-\Delta^{-1})^{-1}]^{-1} \phi(u) \Omega \rangle & \text{if $\beta > \Im(u)>0.$}
\end{cases}
\een
Since $\Delta>0$, $F$ is well defined, in the distributional sense if $s>0$, and we refrain from using a different symbol for $F$ even though it is a different quantity compared 
to  sec. \ref{sec:fermions}. The connection to the modular Hamiltonian now follows from the elementary formula 
\ben
\log \Delta = \int_0^\infty \dd s \, \frac{1}{1-e^s(1-\Delta^{-1})^{-1}} , 
\een
giving 
\ben
\label{bosemod}
\langle \Omega | \phi(u) (\log \Delta) \phi(v) \Omega \rangle = \int_0^\infty \dd s \, F(s,u,v). 
\een
Thus, we should determine $F$. 
Again, the key result is a lemma expressing a jump condition. 

\begin{lemma}\label{lemma19}
Let $v \in A$ be fixed. If $u \in A'$, then 
\ben
F(s, u-i0, v) = F(s, u+i0, v). 
\een
If $u \in A$, then 
\ben\label{jump00}
(1-e^s) F(s, u-i0, v) - F(s, u+i0, v) = \langle \Omega | [ \phi(u), \phi(v) ] \Omega \rangle,  
\een
in the sense of distributions. Here $[ \phi(u), \phi(v) ] = \phi(u) \phi(v)-\phi(v)\phi(u)$ is the commutator. 
\end{lemma}
\begin{proof}
Analogous to the proof of lemma \ref{lemma17}.
\end{proof}

Our assumptions \ref{ass2} and the bosonic nature of $\phi$  imply the operator expansion
\ben
[ \phi(u), \phi(v) ] = \sum_{n=0}^{2h-1} \delta^{(n)}(u-v) O_n(\half(u+v)), 
\een
which may be used to eliminate the commutator in the jump condition \eqref{jump00}. Certain obvious changes also apply 
to the periodicity of $F$, which is now \eqref{period} without the ``$-$'' sign. In the proof of lemma \ref{lemma15} we now use 
the formula 
\ben
\label{jump12}
\frac{1}{1-e^s(1-\Delta^{-1})^{-1}} = 1 - \frac{i}{2} \frac{1}{e^s-1} \int_{-\infty}^\infty \dd t  \frac{[(e^s-1)\Delta]^{it}}{\sinh[\pi(t-i0)]}.
\een
instead of \eqref{jump11}, leading now to $F(s,u-i0,v) \sim \frac{e^{-i\pi h}}{2\pi } (u-v-i0)^{-2h}$ for $u$ near $v$. Altogether, this gives the 
following variant of thm. \ref{thm3} for bosonic fields:

\begin{theorem}\label{thm4}
\begin{itemize}

\item
For fixed $v \in A$ and $s > 0$, 
$F(s, u,v)$ is a holomorphic function of $u$ on the periodically cut complex plane $\CC \setminus (A + \ZZ + \tau \ZZ)$
with the periodicity $F(s, u+1, v) = F(s,u+\tau, v) = F(s,u,v)$.
\item
Across the cuts, $u \in A$, $F$ satisfies the jump condition 
\ben
\label{eq:jumpboson}
(1-e^s) F(s, u-i0, v) - F(s, u+i0, v)=\sum_{n=0}^{2h-1} \delta^{(n)}(u-v) \langle O_n \rangle.
\een

\item
Near the end-points $q_i$ of the intervals, $|F| \lesssim |q_i-u|^{-h}$, as a function of $u$.

\item
Near $v$, 
\ben
\label{pole}
F(s, u- i0, v) \sim \frac{e^{-i\pi h}}{2\pi }(u-v-i0)^{-2h} 
\een
as a function of $u$.

\item
Analogous properties hold true for $u \leftrightarrow v$. 
\end{itemize}

\end{theorem}

\begin{remark} We will see below  that the theorem, or alternatively the following cor. \ref{cor:3}, may be used to find $F$, and thereby 
the modular Hamiltonian in view of \eqref{bosemod}. See sec. \ref{sec:U1current} for examples. 
\end{remark}

In the case of bosonic fields, it is apparently not straightforward to get an analog of the integral equation in cor. \ref{cor:2}. 
The problem is that $G$, the analogue of \eqref{fermitwopoint1} should have a simple pole at $u=v$ and be such that 
$G(u-u')F(s,u',v)$ is doubly periodic in $u'$ for fixed $v \in A$. Since, by contrast to fermionic fields, $F(s,u,v)$ is itself 
doubly periodic in $u$, the simplest choice -- in the absence of any other structural properties of $F$, would be to choose
$G$ as doubly periodic, too. However, as is well known there is no such meromorphic function with only one simple pole in the fundamental 
parallelogram. We will, however, be able to make a similar construction below in the case of the $U(1)$ current, where 
$F$ has an additional property. 

In the general case, we can nevertheless still find an integral equation in the vacuum case, where we can simply set 
$G(x-y) = (2\pi i)^{-1}(x-y)^{-1}$. Here, we go back to the real line picture in which the theory is living on $\RR$ parameterized by $x,y$, and 
we correspondingly write $F(s,x,y)$, etc.
Thm. \ref{thm4} still applies to the vacuum case: By inspection of the proof, we can take in the end the limit $\beta \to \infty$. 
A similar type of argument as for cor. \ref{cor:2} now leads to (using $\langle O_n \rangle = 0$ except when $O_n=1$ in the vacuum state):

\begin{corollary}\label{cor:3}
$F$ defined as in \eqref{Fdef1} (for a bosonic primary field $\phi$ of dimension $h$ in 
standard nomalization \eqref{ddt}) satisfies for the vacuum state $|\Omega\rangle \equiv|\Omega_0\rangle$:
\ben
\begin{split}
\langle \Omega_0|\phi(x) \phi(y) \Omega_0\rangle
=& F(s,x,y+i0) \\
& + e^s \, \frac{1}{2\pi i} \bint_A \dd y' \, F(s,x,y'+i0)  (y'-y-i0)^{-1}  .
\end{split}
\een
for $x,y \in A \subset \RR$ (lightray picture).
\end{corollary}

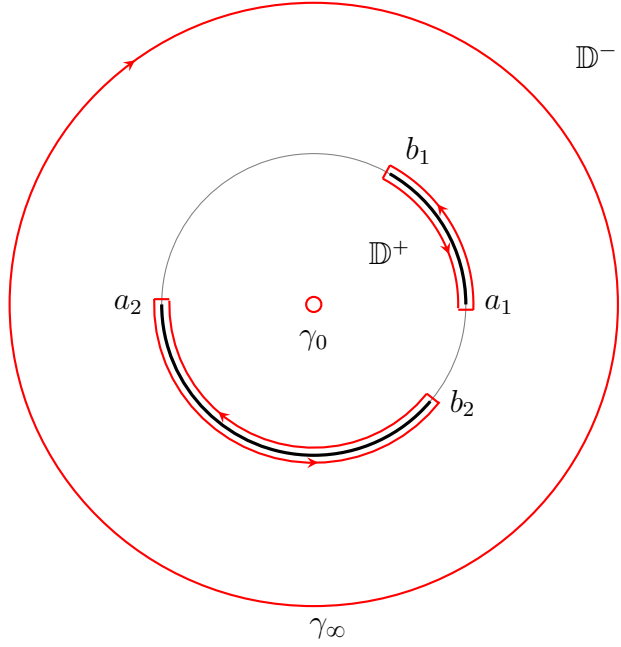
\begin{figure}
\centering
\begin{tikzpicture}[scale=1.0]

\draw[gray] (2,0)  arc[radius = 2, start angle= 0, end angle= 360];
\draw (2.1,0) node[right]{$a_1$}; 
\draw (-2.1,0) node[left]{$a_2$}; 
\draw ({2.1*cos(321.5)},{2.1*sin(321.5)}) node[above,right]{$b_2$};
\draw ({2.3*cos(62.5)},{2.3*sin(62.5)}) node[below,right]{$b_1$};

\draw[black,very thick] (2,0)  arc[radius = 2, start angle= 0, end angle= 60];
\draw[black,very thick] (-2,0)  arc[radius = 2, start angle= 180, end angle= 320];

\draw[directed,red,thick] (4,0)  arc (360:0:4);
\draw[red,thick] (.1,0)  arc (360:0:.1);
\draw (0,-.2) node[below]{$\gamma_0$}; 
\draw (-.2,-4.3) node[right]{$\gamma_\infty$}; 

\draw [red,thick,directed, domain=-2:61.5] plot ({2.1*cos(\x)}, {2.1*sin(\x)});
\draw [red,thick,directed, domain=61.5:-1.5] plot ({1.9*cos(\x)}, {1.9*sin(\x)});
\draw [red,thick] (1.9,-.07) -- (2.1,-.07);
\draw [red,thick] ({1.9*cos(61.5)},{1.9*sin(61.5)}) -- ({2.1*cos(61.5)},{2.1*sin(61.5)});

\draw [red,thick,directed, domain=178:322] plot ({2.1*cos(\x)}, {2.1*sin(\x)});
\draw [red,thick,directed, domain=321.5:178.5] plot ({1.9*cos(\x)}, {1.9*sin(\x)});
\draw [red,thick] (-1.9,.07) -- (-2.1,.07);
\draw [red,thick] ({1.9*cos(321.5)},{1.9*sin(321.5)}) -- ({2.1*cos(321.5)},{2.1*sin(321.5)});

\draw (1,1) node[below]{${\mathbb D}^+$}; 
\draw (3.3,3.3) node[right]{${\mathbb D}^-$}; 

\end{tikzpicture}
\caption{Illustration of the contour deformation in the complex $y'$-plane.}
\label{fig:2}
\end{figure}

\begin{proof}
For the proof, we first work in the circle picture where $x,y \in \SS$. For our assumptions about the fields, we have
for $z \in \bD^+ = \{ z: |z|<1 \}$ that $\phi(z)|\Omega_0\rangle = \sum_{n \ge h} z^{n-h} \phi_n |\Omega_0\rangle$ 
which is holomorphic in $\bD^+$. Likewise, for $z  \in \bD^- = \{ z: |z|>1 \}$ that $\langle\Omega_0| \phi(z) = \sum_{n \ge h} z^{-n-h} \langle \phi_n \Omega_0|$
which is holomorphic in $\bD^-$ and goes as $|z|^{-2h}$ for $|z| \to \infty$. This implies that, for fixed $x \in A$, $F(s,x,y)$ is a holomorphic 
function in $\CC \setminus A$ decaying like $|y|^{-2h}$ for $|y| \to \infty$. Now we consider $y \in \bD^+$, a small contour $\gamma_0$ around $y'=0$, and the identity
\ben
\frac{1}{2\pi i} \oint_{\gamma_0} \dd y' \, F(s,x,y')(y'-y)^{-1} =0. 
\een
We then move the contour $\gamma_0$ across the cuts $A$ as in fig. \ref{fig:2}, deform it to a very large circle $\gamma_\infty$, 
use the decay of $F(s,x,y')$ as well as the jump conditions
\eqref{eq:jumpboson} across the cuts $A$ and collect the residue. The statement then follows after transforming back to the lightray picture. 
\end{proof}

\subsection{Example: Modular Hamiltonian of free $U(1)$-current}
\label{sec:U1current}

The $U(1)$ current, $j$, was introduced in example 6 above. It has dimension $h=1$ and commutator $[j(u),j(v)]=i\delta'(u-v)$. The thermal 2-point 
function is 
\ben
\label{current2point}
\langle \Omega_\beta | j(u) j(v) \Omega_\beta \rangle = -\frac{1}{2\pi} \left( \wp(u-v-i0; \tau) - \eta_1(\tau) \right) \equiv G(u-v-i0). 
\een 
Here, $\eta_1, \eta_2$ are the constants appearing in connection with the Weierstrass $\wp$-function, see appendix. 

Guessing a -- hopefully unique -- answer for $F$ just from the properties given 
in thm. \ref{thm4} as in the case of the free Fermi field does not seem as straightforward for the current, so we proceed by the 
more deductive method of integral equations. We have already mentioned that for thermal states, 
we have not found a general way to obtain an analogue of cor. \ref{cor:2} for bosons. However, for the case of the U(1)-current, 
there is an additional structural property which helps. 

Consider the Weierstrass $\zeta$-function satisfying $\zeta'=-\wp$, see appendix. It has a single simple pole 
at $u=0$ in the fundamental parallelogram
but is only quasi-periodic, $\zeta(u+1)=\zeta(u) + \eta_1, \zeta(u+\tau)=\zeta(u;\tau) +\eta_2$. As a consequence, 
the function $\zeta(u)- \eta_1 u$ is periodic under $u \to u+1$, and changes by $\eta_2  - \eta_1 \tau=-2\pi i $ under 
$u \to u+\tau$.
The combination $\{ \zeta(v'-v) - (v'-v)\eta_1\} F(\xi,u,v')$ is also not doubly periodic in $v'$ for any fixed $u\in A$. However, integrating around a contour $\gamma_\square$ surrounding the fundamental 
parallelogram as in fig. \ref{fig:1}, 
\ben
\label{cancellation}
0=\oint_{\gamma_\square} \dd u \,  F(s,u,v') \{ \zeta(v'-v) - (v'-v)\eta_1\}, 
\een
still gives zero because of the special property of $F$ that, for $u \in A$, 
\ben
\label{cancellation1}
\begin{split}
\int_0^1 F(s,u,v'+\epsilon) \dd v'
=& 
\begin{cases}
\langle [1-e^s (1-\Delta)^{-1}]^{-1} j(u) \Omega |  j_0 \Omega \rangle & \text{if $-\beta < \Im(\epsilon)<0$,}\\
\langle j_0 \Omega | [1-e^s(1-\Delta^{-1})^{-1}]^{-1} j(u) \Omega \rangle & \text{if $\beta > \Im(\epsilon)>0,$}
\end{cases}
\\
=& \ 0,
\end{split}
\een
which follows from the mode expansion of $j(v)=\sum_{n =-\infty}^\infty  j_n e^{-2\pi i n v}$ and the fact that $j_0$ is the charge operator, which annihilates any 
state in the vacuum sector. Hence also $j_0|\Omega\rangle = 0$ for our thermal state, since it can be viewed as a statistical operator on the vacuum Hilbert space.

If we now evaluate this contour integral as before taking the properties of $F$ in thm. \ref{thm4} into account, then we get
\ben\label{inteqn1}
\begin{split}
&\frac{1}{2\pi} \{ \wp(u-v-i0) - \eta_1 \} = F(s,u,v+i0) \\
&+ \frac{1}{2\pi i} e^s \bint_A \dd v' \,  F(s,u,v'+i0) \{ \zeta(v'-v-i0) - (v'-v) \eta_1 \}.
\end{split}
\een
Here, $\bint$ denotes the regulated integral described in and below cor. \ref{cor:2}.

We now consider the connection with formulae in the literature for the 
kernel of the modular Hamiltonian. First, we decompose 
\ben
\label{Gdecomp}
G(u-v-i0) = S(u-v) + \frac{i}{2}C(u-v), 
\een
with $\frac{i}{2} C(u-v)=\delta'(u-v)$ symmetric (imaginary) and $S(u-v)=-\frac{1}{2\pi} {\rm P.V.}  \left( \wp(u-v) - \eta_1\right)$ the symmetric (real) part of $G$ as in \eqref{current2point} (P.V. means Cauchy principal value).
Since the correlation function satisfies $\langle \Omega | j(f)^* j(f) \Omega \rangle \ge 0$, we get $G(\overline f, f) \ge 0$ for any test-function $f \in C^\infty_0(A, \CC)$. Decomposing $f=g+ih$ into a real and imaginary part, we get 
\ben 
\label{eq:cauchy}
\half |C(h,g)| \le \sqrt{S(h,h)}\sqrt{S(g,g)}
\een
for all $g,h \in C^\infty_{0}(A,\RR)$. We then use the positive definite 
$\RR$-bi-linear form $S$ to define on $C^\infty_{0}(A, \RR)$ an inner product, which we then extend to a hermitian inner product by complex anti-linearity in the first entry. 
Let this hermitian positive definite sesqui-linear form be called $( \ , \ )_S$. Its completion defines a Hilbert space, 
${\mathcal K}_A$, which is contained in the Sobolev space $W_0^{1/2,2}(A)$. It is the 1-particle space of the GNS-Fock 
space built on $|\Omega\rangle \equiv |\Omega_\beta \rangle$. It follows from the definitions 
that ${\mathcal K}_{A} \owns f \mapsto \phi(f) |\Omega \rangle \in {\mathcal H}$ is an isometry. 
Furthermore, since $[1-e^s (1-\Delta^{\pm 1})^{-1}]^{-1}$ is a bounded, self-adjoint operator on ${\mathcal H}$ for $s>0$, it follows from \eqref{Fdef1} 
that the kernel $F^\pm(s,x,y)=F(s,x\pm i0,y)$ extends to a bounded quadratic form on ${\mathcal K}_{A}$, and hence can be identified with a bounded linear 
operator on ${\mathcal K}_{A}$.

By Riesz' theorem, there is a self-adjoint operator on this Hilbert space, $\Sigma_A$, such that $\frac{i}{2} C(f,g) = (f, \Sigma_A g)_S$ for all $f,g \in {\mathcal K}_A$.
This operator satisfies $| \Sigma_A | \le 1$, so $| \Sigma_A^{-1}| \ge 1$ in view of \eqref{eq:cauchy}. It follows that for $s>0$, $(e^{-s} - \half + \half \Sigma^{-1}_A)^{-1}$ exits and is a bounded, self-adjoint operator on ${\mathcal K}_A$. By construction, the operator $\Sigma_A$ is expressible as $(i/2) S_A^{-1} C_A$, with $S_A,C_A,G_A$ the operators defined by the 
restriction of the kernels to $A$. Alternatively,
\ben
-iC_A^{-1}G_A^{} = \half (\Sigma^{-1}_A-1).
\een
Using $C=\delta'$ and following the same argument as in sec. 4 of \cite{Casini1} (where the vacuum state 
was considered), we find that 
\ben\label{CAinv}
C_A^{-1}f(u) = \half \int_a^u f(u') \dd u' - \half \int_u^b f(u') \dd u'
\een
in the case of one interval $A=(a,b)$, leading to
\ben
\begin{split}
2\pi C_A^{-1}G_A^{} 
= \zeta(u-v-i0) - u\eta_1-\half[
\zeta(a-v-i0) +\zeta(b-v-i0) - (a+b)\eta_1
]. 
\end{split}
\een
Note that the terms in $[\dots]$ do not depend on $u$, and therefore could be added in the integral \eqref{cancellation} in view of \eqref{cancellation1}. Therefore, we can write our 
integral equation \eqref{inteqn1} in operator form also as:
\ben
F^+_A(s) [1 + \half e^s(\Sigma^{-1}_A-1) ]  =- G_A
\een
where $F^+_A(s)$ is the bounded operator on ${\mathcal K}_{A}$ corresponding to the kernel $F(s,u,v+i0)$. 
Since we already know that the operator on the right side can be inverted for $s>0$, it follows that 
\ben
\label{Fsoln}
F^+_A(s) = -G_A [1 + \half e^s(\Sigma^{-1}_A-1) ]^{-1}
\een
is a solution to our integral equation in operator notation. Integrating this operator identity over $s$ as in eq. \eqref{bosemod}, we obtain 
\ben
\begin{split}
\int_0^\infty \dd s \, F^+_A(s) =& -G_A \int_0^\infty \dd s \, [1 + \half e^s(\Sigma^{-1}_A-1) ]^{-1} \\
=& \ G_A \int_0^1 \dd \lambda \, [\lambda + \half (\Sigma^{-1}_A-1) ]^{-1} \\
=& \ G_A  \log \left[ \frac{\Sigma^{-1}_A + 1}{\Sigma^{-1}_A-1} \right],
\end{split}
\een
and thereby in view of eq. \eqref{bosemod}
\ben
\label{peschl2}
\langle \Omega | j(u) (\log \Delta) j(v) \Omega \rangle = \left(G_A   \log \left[ \frac{\Sigma^{-1}_A + 1}{\Sigma^{-1}_A-1} \right] \right) (u,v). 
\een
This formula is sometimes quoted in the literature \cite{PE} for the matrix elements of the modular Hamiltonian of a free boson field, and should be thought of as analogous 
to \eqref{peschl1} for the fermion field. A similar derivation also goes through if $A=\cup_i (a_i,b_i)$ consists of an arbitrary number of intervals.  By contrast to 
the results in the literature, our derivation is based on the rigorously derived integral equations, rather than a formal analog with finite-dimensional quantum systems.

We now sketch how to find an explicit expression for the matrix element \eqref{peschl2} of the modular Hamiltonian in the case of one interval. 
To calculate the logarithm of operators, we need the spectral decomposition of the self-adjoint operator $\Sigma_A^{-1}$, which we write as
\ben
\Sigma_A^{-1} = -\int_{-\infty}^\infty \dd s \coth(s) \, U_s(u)  V_s(v)
\een
Here, $V_s$ is a generalized\footnote{
More precisely, $V_s$ is a linear form on a domain ${\mathcal D} \subset {\mathcal K}_A$ defined by $g \mapsto \bint_A V_s(u) g(u) \dd u$. 
}
left-eigenfunction of $\Sigma_A^{-1}$ with eigenvalue $V_s \Sigma_A^{-1}  = -\coth(s)  V_s$, 
and $U_s$ satisfies $\bint_a^b V_s(u) U_{s'}(u) \dd u = \delta(s-s')$. 
Paralleling the steps in sec. 4 of \cite{Casini1}, one can see that ($q_i$ ranging over the endpoints $a,b$ of the interval)
\ben\label{Vdef}
V_s(u) = \sum_{q_i}
\frac{c_{q_i}(s) \vartheta_1(u-q_i+is(b-a)/2\pi)}{\vartheta_1(u-q_i)}\left( -\frac{\vartheta_1(u-a) }{\vartheta_1(u-b) } \right)^{-is/2\pi}
\een
with boundary value prescription $u+i0$ understood, is such a generalized eigenfunction. 
This eigenfunction is constructed in such a way that is has a multiplicative jump by $e^s$ across the interval $(a,b)$ (from the last term), 
and such that it is doubly periodic in $u$ (cf. table \ref{tab:1}). Furthermore, near any one $q_i$ of the endpoints $a,b$, it is bounded by $\lesssim |u-q_i|^{-1}$ in accordance 
with thm. \ref{thm4}. The constants $c_{q_i}(s)$ are at first adjusted so that 
\ben
\label{eq:meanvalue}
\bint_0^1 V_s(u-i0) \, \dd u = 0.  
\een
Then, we have again the relation \eqref{cancellation} with $V_s$ in place of $F$, and evaluating this contour integral by deforming the contour around the fundamental 
parallelogram to a tight contour around $(a,b)$ as in fig. \ref{fig:1} then gives the desired eigenvalue equation. That \eqref{eq:meanvalue} impose no actual loss of 
generality is demonstrated in the end. 

Property \eqref{eq:meanvalue} is evidently equivalent to 
\ben
\sum_{q_i} c_{q_i}(s) I_{q_i}(s) =0, 
\een
where 
\ben
I_q(s) =\bint_a^b \dd u \,  \frac{\vartheta_1(u-q+is(b-a)/2\pi)}{\vartheta_1(u-q)} \left( -\frac{\vartheta_1(u-a) }{\vartheta_1(u-b) } \right)^{-is/2\pi} .
\een
The constants  $c_{q_i}(s)$ can further be adjusted so that 
$-i {\rm sgn}(s) C^{-1}_A \overline V_s = U_s $  
is a right-eigenvector of $\Sigma_A^{-1}$, satisfying the desired normalization
$\bint_a^b V_s(u) U_{s'}(u) \dd u = \delta(s-s')$, see 
sec. 3 of \cite{Casini1} for further discussion of this point. The $\delta$-function in this expression can only come from the the contribution of the integrals
near the boundary points $q_i$, which gives a practical way of evaluating these constraints. 

Using $\log[(\coth(s)-1)/(\coth(s)+1)] = -2s$ and $\Sigma_A^{-1} U_s = -\coth(s) U_s$,
 it immediately follows that 
\ben
\label{eq:modham}
-\frac{i}{2} \log \left[ \frac{\Sigma^{-1}_A + 1}{\Sigma^{-1}_A-1} \right] C_A^{-1} =  \int_{-\infty}^\infty \dd s \, |s| \,  U_s(u)    \overline{U_s(v)}  =: k_A(u,v) \ .
\een
Therefore in view of \eqref{peschl2} and the commutation relation $[j(u),j(v)]=i\delta'(u-v)$, the final answer may also be (formally) rewritten as
\ben
\label{eq:modham1}
\langle \Omega | j(u) (\log \Delta, j(v) \Omega \rangle = \langle \Omega | j(u) [H_A, j(v)] \Omega \rangle, 
\een
 where $H_A= \int_{A \times A} k_A(u,v) j(u) j(v) \dd u \dd v$ and $k_A(u,v)$ the kernel of the operator on the right side of \eqref{eq:modham}.
 We refrain here form analyzing in detail the remaining integrals.
 
 This construction can be generalized to the case of $p$ intervals $A=\cup_{j=1}^p (a_j, b_j)$ in the following way. The ansatz for the generalized eigenfunction is now
 \ben\label{Vdefp}
V_s(u) = \sum_{q_i}
\frac{c_{q_i}(s) \vartheta_1(u-q_i+is|A|/2\pi)}{\vartheta_1(u-q_i)}\left( -\frac{\prod_j \vartheta_1(u-a_j) }{\prod_j \vartheta_1(u-b_j) } \right)^{-is/2\pi} ,
\een
where $|A|=\sum_j (b_j-a_j)$.
On the $2p$ coefficients $c_{q_i}(s)$ we impose \eqref{eq:meanvalue} and the $p$ constraints 
\ben
\label{eq:meanvalue1}
\bint_{a_j}^{b_j} V_s(u) \, \dd u = 0, \quad j=1, \dots, p.   
\een 
Of these, only $p-1$ are independent because \eqref{eq:meanvalue} and the periodicity of $V_s(u)$ can be used to show that $\sum_j \bint_{a_j}^{b_j} V_s(u) \, \dd u = 0$
by integrating $V_s(u)$ around a contour as in fig. \ref{fig:1}.
A priori, \eqref{eq:meanvalue1}, \eqref{eq:meanvalue1} are again a restriction on the possible set of eigenfunctions. But in the end one shows that, in fact, all eigenfunctions of $\Sigma^{-1}_A$ satisfy this constraint.  
They are equivalent to 
\ben
\sum_{q_i} c_{q_i}(s) I_{q_i}^{j}(s) =0, \quad j=1, \dots, p, 
\een
where 
\ben
I_q^j(s) =\bint_{a_j}^{b_j} \dd u \,  \frac{\vartheta_1(u-q+is|A|/2\pi)}{\vartheta_1(u-q)} \left( -\frac{\prod_j \vartheta_1(u-a_j) }{\prod_j \vartheta_1(u-b_j) } \right)^{-is/2\pi} .
\een
This leaves us with $p$ linearly independent left-eigenfunctions, $V^k_s(u), k=1, \dots, p$. The constants  $c_{q_i}(s)$ are further be adjusted so that 
$-i {\rm sgn}(s) C^{-1}_A \overline {V_s^k} = U_s^k$  
or equivalently
\ben
U_s^k(u) = -i {\rm sgn}(s) \bint_{a_j}^u \dd u' \, \overline {V_s^k(u')}, \quad \text{for $u\in (a_j,b_j)$,}
\een 
is a right-eigenvector of $\Sigma_A^{-1}$,  satisfying the desired ortho-normalization
$\bint_A V_s^k(u) U_{s'}^{k'}(u) \dd u = \delta_{kk'}\delta(s-s')$. 

Following an argument given in sec. 5 of \cite{Casini1}, let us finally explain why  conditions \eqref{eq:meanvalue},  \eqref{eq:meanvalue1}, initially 
only imposed for convenience in order to derive the desired eigenfunctions, impose no actual loss of generality.   
By expanding the elliptic functions near the end-points of the 
intervals and using the regularization prescription implicit in $\bint$ described in the remark after cor. \ref{cor:2}, that the limits of the functions 
$\lim_{s \to 0} U_s^k(u) = \chi_k(u)$ yields a set of functions whose span is equal to the set of indicator functions $\{1_{(a_j,b_j)}(u)\}_{j=1, \dots, p}$. 
Therefore, the orthogonality relation implies that the $p$ eigenfunctions $\{V_s^k(u) \}_{k=1,\dots, p}$ are already complete, because 
any eigenfunction must satisfy \eqref{eq:meanvalue1}.

Then a similar analysis as in the case of one interval, expanding the bounded operator $F^+(s)$ on ${\mathcal K}_A$ in the generalized basis $U^k_s, V^k_s$
of $\Sigma_A$, gives
again \eqref{eq:modham1} with $H_A = \int_{A \times A} k_A(u,v) j(u) j(v) \dd u \dd v$, where 
\ben
\sum_{k=1}^p  \int_{-\infty}^\infty \dd s \, |s| \,  U^k_s(u)    \overline{U^k_s(v)}  =: k_A(u,v)
\een
now.  Again, we refrain here form analyzing in detail the remaining integrals.

\section{Conclusion}
\label{sec:Conclusion}

In this work we have studied the modular flows of multi-component regions in chiral CFTs using methods from complex analysis. 
Our main tool was the KMS-condition built into modular theory, which we have combined with  input from CFT such as locality and analyticity. 
The main general results concern matrix elements of the modular operator of the general type $\langle \Omega| \phi(x) f(\Delta) \phi(y) \Omega \rangle$, 
where $f$ are various functions such as $\log$, resolvents, complex powers, and $\phi$ a primary field. The states $|\Omega\rangle$ studied in this work were 
the vacuum and Gibbs (thermal) states. In some cases, our results for the matrix elements can be expressed in terms of integral 
equations of Cauchy-type \cite{gakh,mush}, see e.g. cor. \ref{cor:2}, \ref{cor:3}. 

Solving these types of equations was possible in a number of examples, such as free fermions and the $U(1)$-current. To keep the paper at a reasonable length, we have not 
carried out in this work all the resulting integrations, which would be needed if one is interested in more explicit answers. This would be possible. 
In the examples considered, we were also able to explain the relationship with known results in the literature \cite{Casini0, Casini1, blanco, fries}, and also with certain general, but formal (in the QFT context), 
methods known specifically for such free field theories, see \cite{PE,PE1}. In this sense, our methods also serve to make these formal approaches rigorous. 

Our analysis is, in principle, not limited to such free CFTs, and in fact e.g., cor. \ref{cor:2}, \ref{cor:3}, refer to general primary fields in an arbitrary chiral CFT, subject only 
to certain standard assumptions. It would be interesting to study the implications of these general results further, which we leave for future work \cite{HollandsHieu}.

\medskip
\noindent
{\bf Acknowledgements:} Part of this work was carried out while I was visiting IHES, Paris (March 2019), and Instituto di Atomico Balseiro, Bariloche (March 2018). It is a pleasure to thank these institutions for hospitality and  for financially supporting those visits. 
I am grateful to the Max-Planck Society for supporting the collaboration between MPI-MiS and Leipzig U., grant Proj. Bez. M.FE.A.MATN0003. 
Discussions with H. Casini, P. Fries, M. Huerta, F. Otto, and S. Del Vecchio are gratefully acknowledged. 

\appendix

\section{Conventions for elliptic functions}

Our conventions for the $\vartheta$-functions are:
\ben
\begin{split}
\vartheta_1(u; \tau) &=\sum_{n \in {\mathbb Z}} (-1)^{n-\frac12} e^{i\pi \tau (n+\frac12)^2 + 2\pi i (n+\frac12) u}\\
\vartheta_2(u; \tau) &=\sum_{n \in {\mathbb Z}} e^{i\pi \tau (n+\frac12)^2 + 2\pi i (n+\frac12) u} \\
\vartheta_3(u; \tau) &= \sum_{n \in {\mathbb Z}} e^{i\pi \tau n^2 + 2\pi i n u}, \\
\vartheta_4(u; \tau) &= \sum_{n \in {\mathbb Z}} (-1)^n e^{i\pi \tau n^2 + 2\pi i n u}.
\end{split}
\een

\begin{table}[h]
\begin{center}
\begin{tabular}{|c | c | c | c | c |}
\hline
& $\vartheta_1$ & $\vartheta_2$ & $\vartheta_3$ & $\vartheta_4$ \\ \hline \hline
$\vartheta_i(u+1)/\vartheta_i(u)$ & $-1$& $-1$& $1$& $1$  \\
$\vartheta_i(u+\tau)/\vartheta_i(u)$ & $-M$ & $M$ & $M$ & $-M$  \\
\hline
\end{tabular}
\end{center}
\caption{Periodicity properties, with $M=e^{-i\pi \tau} e^{-2\pi i u}$.} \label{tab:1}
\end{table}
\noindent
The $\wp$-function is given by 
\ben
\wp(u;\tau) = \sum_{(m,n)\neq(0,0)} \left(
\frac{1}{(u+m+\tau n)^2} -\frac{1}{(m+\tau n)^2}
\right).
\een
The Weierstrass $\zeta$-function is a function defined so that $\zeta'(u,\tau) = -\wp(u;\tau)$. Explicitly, 
\ben
\zeta(u;\tau) = \frac{1}{u} + \sum_{(m,n)\neq(0,0)} \left(
\frac{1}{u+m+\tau n} -\frac{1}{m+\tau n} + \frac{u}{(m+\tau n)^2}
\right).
\een
It has the periodicity $\zeta(u+1) = \zeta(u)+\eta_1, \zeta(u+\tau)=\zeta(u)+\eta_2$.


\begin{thebibliography}{99}

\bibitem{Longo1} 
  R.~Longo and F.~Xu,
  ``Comment on the Bekenstein bound,''
  J.\ Geom.\ Phys.\  {\bf 130}, 113 (2018)

\bibitem{bousso1} 
  R.~Bousso, Z.~Fisher, S.~Leichenauer and A.~C.~Wall,
  ``Quantum focusing conjecture,''
  Phys.\ Rev.\ D {\bf 93}, no. 6, 064044 (2016)


\bibitem{bousso2} 
  R.~Bousso, Z.~Fisher, J.~Koeller, S.~Leichenauer and A.~C.~Wall,
  ``Proof of the Quantum Null Energy Condition,''
  Phys.\ Rev.\ D {\bf 93}, no. 2, 024017 (2016)


    \bibitem{Casini2} 
  H.~Casini and M.~Huerta,
 `` A c-theorem for the entanglement entropy,''
  J.\ Phys.\ A {\bf 40}, 7031 (2007)
  
    \bibitem{Hubeny}
M. Rangamani and T. Takayanagi, {\it Holographic Entanglement Entropy}, Springer Lecture Notes in Physics (2017)
  
  \bibitem{Takesaki}
  M. Takesaki {\it Theory of operator algebras, I-III}, Springer (2003)
  
  \bibitem{bisognano}
   J.~J.~Bisognano and E.~H.~Wichmann,
  ``On the Duality Condition for Quantum Fields,''
  J.\ Math.\ Phys.\  {\bf 17}, 303 (1976)
  
    \bibitem{hislop}
P.~D.~Hislop and R.~Longo,
  ``Modular Structure of the Local Algebras Associated With the Free Massless Scalar Field Theory,''
  Commun.\ Math.\ Phys.\  {\bf 84}, 71 (1982)

  \bibitem{longo3}
R.~Brunetti, D.~Guido and R.~Longo,
  ``Modular structure and duality in conformal quantum field theory,''
  Commun.\ Math.\ Phys.\  {\bf 156}, 201 (1993)
  
   \bibitem{Casini0}
  H.~Casini and M.~Huerta,
  ``Reduced density matrix and internal dynamics for multicomponent regions,''
  Class.\ Quant.\ Grav.\  {\bf 26}, 185005 (2009)
  
  \bibitem{Casini1}
  R.~E.~Arias, H.~Casini, M.~Huerta and D.~Pontello,
  ``Entropy and modular Hamiltonian for a free chiral scalar in two intervals,''
  Phys.\ Rev.\ D {\bf 98}, no. 12, 125008 (2018)

  \bibitem{HHW}
  R.~Haag, N.~M.~Hugenholtz and M.~Winnink,
  ``On the Equilibrium states in quantum statistical mechanics,''
  Commun.\ Math.\ Phys.\  {\bf 5}, 215 (1967)
    
  \bibitem{Bratteli}
 O.~Bratteli and D.~W. Robinson,
\newblock {\em {Operator Algebras and Quantum Statistical Mechanics I}}.
\newblock Springer (1987)
O.~Bratteli and D.~W. Robinson.
\newblock {\em {Operator Algebras and Quantum Statistical Mechanics II}}.
\newblock Springer (1997)
  
      \bibitem{sanders_2}
  S.~Hollands and K.~Sanders,
  {\it Entanglement measures and their properties in quantum field theory,}
  SpringerBriefs in Mathematical Physics (2018)
  
    \bibitem{witten} 
  E.~Witten,
  ``Notes on Some Entanglement Properties of Quantum Field Theory,''
  arXiv:1803.04993 [hep-th].
  
  \bibitem{araki_1}
H. Araki,
``Relative Hamiltonian for faithful normal states of a von Neumann algebra,''
 Publ. RIMS Kyoto Univ. {\bf 9}, 165-209 (1973)

   \bibitem{araki_3}
 H. Araki,
``Relative entropy of states of von Neumann algebras.I,II.''
Publ. RIMS Kyoto Univ. {\bf 11}, 809-833 (1976) and {\bf 13}, 173-192 (1977)

     \bibitem{haag}
  R.~Haag,
  {\it Local quantum physics: Fields, particles, algebras,}
  Springer: Berlin (1992)

\bibitem{frohlich}
F. Gabbiani and J. Fr\" ohlich: ``Operator Algebras and Conformal Field Theories,''
Commun. Math. Phys. 155, 569-640 (1993) 

  \bibitem{longo2}
S.~Carpi, Y.~Kawahigashi, R.~Longo and M.~Weiner,
  ``From vertex operator algebras to conformal nets and back,''
  arXiv:1503.01260 [math.OA].

   \bibitem{fredenhagen_5}
  K.~Fredenhagen,
\newblock {``On the modular structure of local algebras of observables},''
\newblock {Commun. Math. Phys.} \textbf{97}, 79--89 (1985)

  \bibitem{buchholz4}
  D.~Buchholz, K.~Fredenhagen and C.~D'Antoni,
  ``The Universal Structure of Local Algebras,''
  Commun.\ Math.\ Phys.\  {\bf 111}, 123 (1987)

\bibitem{buchholz3} 
  D.~Buchholz, G.~Mack and I.~Todorov,
  ``The Current Algebra on the Circle as a Germ of Local Field Theories,''
  Nucl.\ Phys.\ Proc.\ Suppl.\  {\bf 5B}, 20 (1988).
  
   \bibitem{bischoff}
  M.~Bischoff and Y.~Tanimoto,
  ``Construction of wedge-local nets of observables through Longo-Witten endomorphisms. II,''
  Commun.\ Math.\ Phys.\  {\bf 317}, 667 (2013)
  doi:10.1007/s00220-012-1593-x
  
  \bibitem{buchholz}
 D.~Buchholz and H.~Schulz-Mirbach,
  ``Haag duality in conformal quantum field theory,''
  Rev.\ Math.\ Phys.\  {\bf 2}, 105 (1990)
  
  
  \bibitem{goodman_wallach}
R. Goodman and N.R. Wallach, ``Structure and unitary cocycle representations of loop groups and the group of diffeomorphisms of the circle,'' J. Reine Angew. Math. {\bf 347}, (1984) 69-133 

\bibitem{goodman_wallach1}
R. Goodman and N.R. Wallach, ``Projective unitary positive-energy representations of Diff($S^1$),'' J. Funct. Anal. {\bf 63}, 
(1985), no. 3, 299-321

\bibitem{Lang}
S. Lang, {\it $SL_2(\RR)$,} Springer Verlag (1985)

\bibitem{rehren1}
K.~Fredenhagen, K.~H.~Rehren and B.~Schroer,
  ``Superselection sectors with braid group statistics and exchange algebras. 2. Geometric aspects and conformal covariance,''
  Rev.\ Math.\ Phys.\  {\bf 4}, no. spec01, 113 (1992).

\bibitem{araki_5}
H. Araki,
``On quasifree states of the CAR and Bogoliubov automorphisms,''
Publ. RIMS Kyoto Univ. {\bf 6}, 385-442 (1970)


\bibitem{mush}
N. I. Muskhelishvili {\it Singular Integral Equations,} Wolters-Noordhoff Publishing Groningen (1958)

\bibitem{gakh}
F. D. Gakhov, {\it Boundary Value Problems,} Dover (1990), reprint from 1966 edition in Pergamon Press

\bibitem{Rehren:2012wa} 
  K.~H.~Rehren and G.~Tedesco,
  ``Multilocal fermionization,''
  Lett.\ Math.\ Phys.\  {\bf 103}, 19 (2013)

\bibitem{Longo:2009mn} 
  R.~Longo, P.~Martinetti and K.~H.~Rehren,
  ``Geometric modular action for disjoint intervals and boundary conformal field theory,''
  Rev.\ Math.\ Phys.\  {\bf 22}, 331 (2010)

  \bibitem{accardi}
L. Accardi and C. Cecchini, ``Conditional Expectations in von Neumann Algebras and a Theorem of Takesaki,''
J. Funct. Anal. {\bf 45}, 245-273 (1982)

\bibitem{hollands0}
S.~Hollands and R.~M.~Wald,
  ``Existence of local covariant time ordered products of quantum fields in curved space-time,''
  Commun.\ Math.\ Phys.\  {\bf 231}, 309 (2002)
  doi:10.1007/s00220-002-0719-y

  \bibitem{zuber} 
  C.~Itzykson and J.~B.~Zuber,
  ``Two-Dimensional Conformal Invariant Theories on a Torus,''
  Nucl.\ Phys.\ B {\bf 275}, 580 (1986).
  doi:10.1016/0550-3213(86)90576-6

\bibitem{fries}
P.~Fries and I.~A.~Reyes,
  ``The entanglement spectrum of chiral fermions on the torus,''
  Phys.\ Rev.\ Lett.\  {\bf 123}, no. 21, 211603 (2019)
  doi:10.1103/PhysRevLett.123.211603

\bibitem{blanco}
D.~Blanco and G.~Perez-Nadal,
  ``Modular Hamiltonian of a chiral fermion on the torus,''
  Phys.\ Rev.\ D {\bf 100}, no. 2, 025003 (2019)
 doi:10.1103/PhysRevD.100.025003

\bibitem{PE}
I. Peschel, ``Calculation of reduced density matrices from correlation functions,'' J. Phys. A 36, L205 (2003) 
doi:10.1088/0305-4470/36/14/101

\bibitem{PE1}
I. Peschel and V. Eisler, ``Reduced density matrices and entanglement entropy in free lattice
models,'' J. Phys. A: Math. Theor. 42, 504003 (2009)

\bibitem{HollandsHieu}
S. Hollands, in progress.
  
  \end{thebibliography}
\end{document}